\documentclass[11pt,a4paper,reqno]{article}
\usepackage[utf8]{inputenc}
\usepackage{geometry}
\usepackage{mathrsfs}
\usepackage{amsthm,amsmath}
\usepackage{tikz}
\usepackage{float}
\usepackage{enumitem}
\usepackage{dsfont}
\usepackage{mathtools}
\usepackage{subcaption}
\usepackage{url}
\usepackage{amsmath, bm}
\usepackage{amssymb}
\usepackage{xfrac}
\usepackage{pgfplots}
\usepackage{cancel}
\usepackage{tensor}
\usepackage{siunitx}
\usepackage[nottoc]{tocbibind}
\usepackage{xcolor}
\usepackage{hyperref}
\numberwithin{equation}{section}
\newcommand*\diff{\mathop{}\!\mathrm{d}}

\newgeometry{vmargin={20mm}, hmargin={20mm,20mm}}
\pgfplotsset{width = 10cm, compat = 1.9}

\newtheorem*{lemma}{Lemma}

\DeclareSymbolFont{yhlargesymbols}{OMX}{yhex}{m}{n} \DeclareMathAccent{\yhwidehat}{\mathord}{yhlargesymbols}{"62} 
\newcommand{\coin}[1]{\left[\!\left[#1\right]\!\right]}
\DeclareMathOperator{\cosech}{cosech}
\DeclareMathOperator{\sinc}{sinc}
\DeclareMathOperator{\sinhc}{sinhc}
\DeclareMathOperator{\Span}{span}
\DeclareMathOperator{\WF}{WF}

\DeclareMathOperator{\sgn}{sgn}
\newcommand{\kb}{\boldsymbol{k}}
\newcommand{\xb}{\boldsymbol{x}}

\usepackage{scalerel}
\usetikzlibrary{svg.path}

\definecolor{orcidlogocol}{HTML}{A6CE39}
\tikzset{
	orcidlogo/.pic={
		\fill[orcidlogocol] svg{M256,128c0,70.7-57.3,128-128,128C57.3,256,0,198.7,0,128C0,57.3,57.3,0,128,0C198.7,0,256,57.3,256,128z};
		\fill[white] svg{M86.3,186.2H70.9V79.1h15.4v48.4V186.2z}
		svg{M108.9,79.1h41.6c39.6,0,57,28.3,57,53.6c0,27.5-21.5,53.6-56.8,53.6h-41.8V79.1z M124.3,172.4h24.5c34.9,0,42.9-26.5,42.9-39.7c0-21.5-13.7-39.7-43.7-39.7h-23.7V172.4z}
		svg{M88.7,56.8c0,5.5-4.5,10.1-10.1,10.1c-5.6,0-10.1-4.6-10.1-10.1c0-5.6,4.5-10.1,10.1-10.1C84.2,46.7,88.7,51.3,88.7,56.8z};
	}
}

\newcommand\orcidicon[1]{\href{https://orcid.org/#1}{\mbox{\scalerel*{
				\begin{tikzpicture}[yscale=-1,transform shape]
					\pic{orcidlogo};
				\end{tikzpicture}
			}{|}}}}

\title{Quantum Energy Inequalities along stationary worldlines}
\author{Christopher J. Fewster\,\orcidicon{0000-0001-8915-5321} ${}^{1,}$\thanks{\tt chris.fewster@york.ac.uk}
~and Jacob Thompson\,\orcidicon{0000-0003-1047-9257}${}^{1,2,}$\thanks{\tt jthompson16@sheffield.ac.uk} \\ 
${}^1$\,\emph{\small Department of Mathematics,
	University of York, Heslington, York YO10 5DD, United Kingdom.}
\\
${}^2$\,\emph{\small  School of Mathematics and Statistics, The University of Sheffield, Hicks Building, Hounsfield Road,}
\\
\emph{\small  Sheffield S3 7RH, United Kingdom.}}
\date{\today}

\begin{document}

\maketitle

\begin{abstract}
    Quantum energy inequalities (QEIs) are lower bounds on the averaged energy density of a quantum field. They have been proved for various field theories in general curved spacetimes but the explicit lower bound is not easily calculated in closed form. In this paper we study QEIs for the massless minimally coupled scalar field in four-dimensional Minkowski spacetime along stationary worldlines -- curves whose velocity evolves under a $1$-parameter Lorentz subgroup -- and find closed expressions for the QEI bound,
    in terms of curvature invariants  of the worldline. Our general results are illustrated by specific computations for the six prototypical stationary worldlines.    
    When the averaging period is taken to infinity, the QEI bound is consistent with a constant energy density along the worldline.  For inertial and uniformly linearly accelerated worldlines, this constant value is attained by the Minkowski and Rindler vacuums respectively. It is an open question as to whether the bounds for other stationary worldlines are attained by other states of interest. 
\end{abstract}

\section{Introduction}

Even if a classical field theory obeys local energy conditions, such as positivity of energy density, the
corresponding quantum field theory (QFT) will fail to do so, as a result of a general theorem~\cite{EpsteinGlaserJaffe:1965}.
In fact, it is typical that the expectation value of energy density at any given point can be made arbitrarily negative by a suitable choice of the quantum state~\cite{Fews05}. 
This surprising fact is deeply related to the uncertainty principle. As unrestricted negative energy densities, or negative energy fluxes, could produce
a range of effects, ranging from violations of the second law of thermodynamics~\cite{Ford1978QuantumCE} and evasion of the classical singularity theorems of Penrose~\cite{Penrose:1965} and Hawking~\cite{Hawking:1966}, to the ability to construct warp drives~\cite{Alcubierre_1994} or wormholes~\cite{MorrisThorne:1988}, it is important to understand what restrictions might be imposed by QFT itself. In this paper we will be concerned with \emph{Quantum Energy Inequalities} (QEIs), which provide lower bounds on local averages of the expected energy density, independent of the quantum state.
  
Starting with results of Ford and Roman~\cite{Ford1978QuantumCE,Ford:1991,FordRoman:1995} QEIs have been derived for a variety of quantum fields in flat and curved spacetimes. References and discussion may be found in the recent reviews~\cite{Fewster2017,KontouSanders:2020}.
At the simplest level, the QEIs constrain the magnitude and duration of violations of classical energy conditions, placing stringent constraints on attempts to use quantum fields to provide the exotic matter required to construct exotic spacetimes~\cite{FordRoman:1996worm,PfenningFord:1997,FewsterRoman:2005}. It has also been shown
that QEI restrictions are sufficient for modified versions of singularity theorems to hold~\cite{FewsterKontou2022}.

For example, consider the real scalar field of mass $m\ge 0$ in any globally
hyperbolic spacetime $(M,g)$, recalling that global hyperbolicity demands only that
the spacetime possesses a global Cauchy surface. Let $\gamma(s)$ be any smooth
timelike curve, parameterised by proper time. It was shown in~\cite{GenQEI2000CFewster} that the energy density of the quantum field along $\gamma$ obeys the QEI 
\begin{equation}\label{eq:genQEI}
    \int_{-\infty}^\infty \diff{s}   |g(s)|^2\langle\mathbf{:}T_{\mu\nu} \dot{\gamma}^\mu\dot{\gamma}^\nu\mathbf{:}\rangle_\omega(\gamma(s))\ge
    -\int_0^\infty \frac{\diff{\alpha}}{\pi} \widehat{g\otimes g T}(-\alpha,\alpha) >-\infty,
\end{equation}
which holds for all real-valued compactly supported smooth test functions $g$, 
and all Hadamard states $\omega$ of the field. Here, the hat denotes a Fourier transform, defined according to the convention $\hat{g}(\alpha)=\int_{-\infty}^\infty \diff{s} \,e^{i\alpha s}g(s)$, and we employ units where $\hbar=c=1$, which will be in force throughout this paper. 
On the left-hand side, the normal ordering is conducted with respect to an arbitrary Hadamard reference state $\omega_0$, whose two-point function is used to construct the distribution $T(s,s')$ that appears on the right-hand side.  Recall also that the Hadamard states form a large class of physically reasonable states, determined by their short-distance structure~\cite{KAY1991,Moretti:2021}. The two most important features of the QEI~\eqref{eq:genQEI} are that the right-hand side is completely independent of the state $\omega$, and that the bound is finite -- which is proved using the microlocal properties of Hadamard states uncovered by Radzikowski~\cite{Radzikowski:1996}. Discussion of QEIs for other QFTs, including non-free models, may be found in~\cite{Fewster2017,KontouSanders:2020}; see~\cite{FroebCadamuro:2022} for a very recent development.

Although the lower bound in~\eqref{eq:genQEI} is explicit and rigorous, it is not easy to compute in closed form except in special cases. Those examples where explicit calculation is possible are therefore particularly valuable because they can provide insight into the nature of QEI bounds in general.
To the best of our knowledge this has only been achieved when $T$ exhibits translational invariance $T(s+r,s'+r)=T(s,s')$ which occurs, for instance, when $(M,g)$ is a stationary spacetime, $\gamma$ is a timelike Killing orbit and $\omega_0$ is stationary. Translational invariance allows us to write, with an abuse of notation, $T(s,s')=T(s-s')$, from which one easily finds that the QEI~\eqref{eq:genQEI} simplifies to
\begin{equation}\label{eq:genQEI2}
    \int_{-\infty}^\infty \diff{s}   |g(s)|^2\langle\mathbf{:}T_{\mu\nu} \dot{\gamma}^\mu\dot{\gamma}^\nu\mathbf{:}\rangle_\omega(\gamma(s))\ge
    -\int_{-\infty}^\infty \diff{\alpha} |\hat{g}(\alpha)|^2 Q(\alpha), 
\end{equation}
where
\begin{equation} \label{eq:Qboundinitial}
    Q(\alpha) = \dfrac{1}{2\pi^2}\int_{-\infty}^\alpha \diff{u} \ \hat{T}(u);
\end{equation}
the QEI~\eqref{eq:genQEI2} is also valid for complex-valued $g$. Taking the massless free field as an example, averaging along an inertial worldline in Minkowski space and using the Minkowski vacuum as the reference state $\omega_0$, this results in $Q(\alpha)=\alpha^4\Theta(\alpha)/(16\pi^3)$. Using the evenness of $|\hat{g}|^2$ together with Parseval's theorem then yields
\begin{equation}\label{eq:genQEI3}
    \int_{-\infty}^\infty \diff{s}   |g(s)|^2\langle\mathbf{:}T_{\mu\nu} \dot{\gamma}^\mu\dot{\gamma}^\nu\mathbf{:}\rangle_\omega(\gamma(s))\ge
    -\frac{1}{16\pi^2}\int_{-\infty}^\infty \diff{s} |g''(s)|^2 .
\end{equation}
Similar expressions are known for massive fields and in Minkowski spacetime of general dimension~\cite{BoundNegEnDenFewsterEveson1998}; for some curved spacetime examples see~\cite{FewsterTeo:1999}. Another explicit example arises where $\gamma$ is a uniformly linearly accelerated worldline in four-dimensional Minkowski spacetime with proper acceleration $a$, in which case the QEI~\eqref{eq:genQEI2} becomes~\cite{QEI2006FewsterPfenning} 
\begin{equation} \label{eq:QEIlinacc}
    \int_{-\infty}^\infty \diff{s} |g(s)|^2\langle \mathbf{:}T_{\mu\nu} \dot{\gamma}^\mu\dot{\gamma}^\nu\mathbf{:}\rangle_\omega(\gamma(s)) \geq -\dfrac{1}{16\pi^2} \int_{-\infty}^\infty \diff{s} \left(\lvert g''(s) \rvert^2 + 2a^2\lvert g'(s) \rvert^2 + \dfrac{11}{30}a^4 \lvert g(s) \rvert^2\right),
\end{equation}
and is again valid for all Hadamard states $\omega$ and complex-valued test functions $g$.

Using such expressions the scaling behaviour of the bound is easily understood and phenomena such as `quantum interest' may be explored~\cite{FordRoman:1999,FewsterTeo:2000,Fewster2017}. For example,
let $g_\lambda(s) =\lambda^{-1/2}g(s/\lambda)$, where $g$ is normalised so that $\int_{-\infty}^\infty \diff{s} |g(s)|^2=1$. Then~\eqref{eq:QEIlinacc} implies  
\begin{equation}\label{eq:AWECacc}
    \liminf_{\lambda\xrightarrow{}\infty}\int_{-\infty}^\infty \diff{s} |g_\lambda(s)|^2\langle\mathbf{:}T_{\mu\nu} \dot{\gamma}^\mu\dot{\gamma}^\nu\mathbf{:} \rangle_\omega(\gamma(s)) \geq -\dfrac{11a^4}{480\pi^2},
\end{equation}
reducing to the 
Averaged Weak Energy Condition (AWEC)  
\begin{equation}\label{eq:AWEC}
    \liminf_{\lambda\xrightarrow{}\infty}\int_{-\infty}^\infty \diff{\tau} |g_\lambda(s)|^2\langle\mathbf{:}T_{\mu\nu} \dot{\gamma}^\mu\dot{\gamma}^\nu\mathbf{:} \rangle_\omega(\gamma(s)) \geq 0
\end{equation}
in the limit $a\to 0$, which can also be obtained directly from~\eqref{eq:genQEI3}. An interesting observation is that the lower bound in~\eqref{eq:AWECacc} is exactly the constant energy density of the Rindler vacuum state along the accelerated worldline, while the lower bound in~\eqref{eq:AWEC} is the energy density of the Minkowski vacuum.  

As closed form expressions for QEI bounds are relatively few in number, it is of interest to find others. The purpose of this paper is to present a calculation of the QEI bound for a massless scalar field along any stationary worldline in $4$-dimensional Minkowski spacetime. By a stationary worldline, we mean
any timelike curve $\gamma(s)$, parameterised by proper time $s$, whose velocity vector evolves under a $1$-parameter subgroup of the Lorentz group: $\dot{\gamma}(s) = \exp (sM) \dot{\gamma}(0)$ for some fixed $M\in\mathfrak{so}(1,3)$ and future-pointing unit timelike $\dot{\gamma}(0)$.  The family of stationary worldlines includes physically important examples including inertial, uniformly linearly accelerated and uniformly rotating worldlines. We briefly recall their significance.

Stationary worldlines have a long history in relativity and quantum field theory.
Kottler~\cite{Kottler1912}, Synge~\cite{Synge1967} and Letaw~\cite{Letaw:1981} (see also~\cite{LetawPfautsch:1982}) all obtained them as the solutions to four-dimensional Frenet-Serret equations subject to constancy of the curvature invariants of the worldline; the name `stationary worldlines' is due to Letaw. 
The three curvature invariants are the curvature, which measures the proper acceleration, and the torsion and hypertorsion, which specify its proper angular velocity. More details are given in Section~\ref{sec:stationary}.
Stationary worldlines are equivalently described as the orbits of timelike Killing vector fields in Minkowski spacetime. There are also overlaps with the theory of rigid motions in special relativity that goes back to Born~\cite{Born:1909} and Herglotz~\cite{Herglotz1910}; in particular, any \emph{rotational} rigid motion is the flow of a timelike Killing vector by the Herglotz--Noether theorem, although the same theorem allows any $C^2$ timelike curve to be a flow line of an \emph{irrotational} rigid motion. See~\cite{Giulini:2010} for discussion and references. 
By a Poincar\'e transformation, any stationary worldline can be reduced to one of six prototypes: the
inertial, uniformly linearly accelerated, and uniformly rotating worldlines have already been mentioned, while the three remaining ones have
spatial projections corresponding to a semicubical parabola, a catenary or a helix. We will give more detail as we discuss each case separately later on.

Turning to quantum field theory, the uniformly linearly accelerated trajectory is of course at the heart of the extensive literature on the Unruh effect~\cite{Unruh:1976} describing the thermal excitation of a detector travelling along the accelerated trajectory~\cite{Unruh:1976,DeBievreMerkli:2006}. Coordinates adapted to this trajectory naturally cover the Rindler wedge, to which the Minkowski vacuum state restricts as a thermal state at the Unruh temperature relative to time translation along the accelerated trajectory. This has attracted the attention of other theorists to understand what excitations may be expected for detectors following other stationary worldlines~\cite{Letaw:1981,GoodJuarezAubryetal:2020} (and references therein) and whether there are ground states associated to coordinates based on these worldlines differing from restrictions of the Minkowski vacuum~\cite{LetawPfautsch:1981}. Proposals to observe the Unruh effect, on the other hand, 
focus on circular trajectories because they are confined to a finite volume~\cite{BellLeinaas:1983,Gooding_etal_PRL:2020}. In summary, there is good motivation to deepen the understanding of physical effects seen along stationary worldlines.

The main result of this paper is that the QEI~\eqref{eq:genQEI2} along any stationary worldline in Minkowski spacetime may be given explicitly as
\begin{equation} \label{eq:QEIfinalintro}
    \int \diff{s} |g(s)|^2\langle\mathbf{:}T_{\mu\nu} \dot{\gamma}^\mu\dot{\gamma}^\nu\mathbf{:} \rangle_\omega(\gamma(s)) \geq -\dfrac{1}{16\pi^2}\int_{-\infty}^\infty \diff{s} \left(\lvert g''(s) \rvert^2 + 2A\lvert g'(s) \rvert^2 + B\lvert g(s) \rvert^2\right),
\end{equation} 
where $A$ and $B$ are expressed in terms of the curvature $\kappa$, torsion $\tau$, and hypertorsion $\upsilon$ as 
\begin{equation} \label{eq:curvatureinvA}
    A = \kappa^2 + \tau^2 + \upsilon^2
\end{equation}
and 
\begin{equation} \label{eq:curvatureinvB}
B =\frac{1}{90}\left( 3\kappa^4 + 62\kappa^2\tau^2+30(\kappa^2 + \tau^2 + \upsilon^2)^2\right),
\end{equation}
and the inequality~\eqref{eq:QEIfinalintro} holds for all Hadamard states $\omega$ 
and all smooth compactly supported test functions $g$. 
Because general stationary worldlines describe a richer range of behaviour than that of constant linear acceleration, the above formulae provide a more refined understanding than that given by~\eqref{eq:QEIlinacc} of the way QEI bounds along a general worldline are sensitive to its curvature invariants.

To interpret the QEI~\eqref{eq:QEIfinalintro}, it is useful to consider
its scaling behaviour. As before, we take a test function $g_\lambda$ which is just a scaled version of the test function $g$, namely $g_\lambda(s) = \lambda^{-1/2}g(s/\lambda)$, so the support width of $g_\lambda$ is proportional to $\lambda$. Observing that $g_\lambda^{(k)}(s) = \lambda^{-k-1/2} g^{(k)}(s/\lambda)$, we find 
\begin{equation} \label{eq:scaling} 
    \int \diff{s} |g_\lambda(s)|^2\langle\mathbf{:}T_{\mu\nu} \dot{\gamma}^\mu\dot{\gamma}^\nu\mathbf{:} \rangle_\omega(\gamma(s)) \geq -\dfrac{\|g''\|^2}{16\pi^2\lambda^4} - \dfrac{A\|g'\|^2}{8\pi^2\lambda^2} - T_\text{reg}(0)\|g\|^2,
\end{equation}
where the norms are those of $L^2(\mathbb{R})$, i.e., $\|g\|^2 = \int_{-\infty}^\infty \diff{s} |g(s)|^2$. Here we have written $T_{\text{reg}}(0) = B/(16\pi^2)$ for reasons that will become clear later -- see, for example, equation \eqref{eq:TsplitIntro} and the arguments presented in Section \ref{sec:genQEI}. For sampling times shorter than the curvature scales, i.e., $\lambda\ll \min\{\kappa^{-1},\tau^{-1},\upsilon^{-1}\}$, the leading term dominates, reflecting the fact that any worldline is approximately inertial on short enough timescales. At intermediate and long timescales relative to curvature scales, the bound will receive corrections from, and eventually be dominated by, the last two terms in~\eqref{eq:scaling}, showing that the QEI is sensitive to the curvature invariants of the worldline $\gamma$. In the limit $\lambda\to+\infty$, and with $g$ normalised so that $\|g\|=1$, we obtain the remarkably simple formula
\begin{equation}\label{eq:longtime}
 \liminf_{\lambda\xrightarrow{}\infty}\int_{-\infty}^\infty \diff{s} |g_\lambda(s)|^2\langle\mathbf{:}T_{\mu\nu} \dot{\gamma}^\mu\dot{\gamma}^\nu\mathbf{:} \rangle_\omega(\gamma(s)) \geq -T_\text{reg}(0),
\end{equation} 
which bounds the average energy density along the entire trajectory.
In particular, the QEI is consistent with the existence of a constant renormalised energy density $-T_\text{reg}(0)$ along $\gamma$, and this is the most negative value that any constant energy density could take.  An intriguing question is whether or not this value is attained by some Hadamard state, or a sequence of Hadamard states in a limiting sense, which we will address in Section~\ref{sec:discussion}. 

The derivation of~\eqref{eq:QEIfinalintro} requires a number of innovations. 
Although the point-split energy density can be obtained easily enough for any given stationary worldline, its Fourier transform does not have a closed form -- as far as we know -- for three of the six prototypes. In Section \ref{sec:genQEI}, we develop a new method for computing the QEI bound for massless fields in four-dimensions that avoids the use of Fourier transforms. The result is that the QEI will take the form~\eqref{eq:QEIfinalintro} provided that the point-split energy density takes the form
\begin{equation} \label{eq:TsplitIntro}
    T(s,s')= \lim_{\epsilon\to 0+}\left( \frac{3}{2\pi^2(s-s'-i\epsilon)^4} -\frac{A}{4\pi^2 (s-s'-i\epsilon)^2}\right) + T_{\text{reg}}(s-s'),
\end{equation} 
where the regular part $T_{\text{reg}}$ must satisfy various conditions, whereupon the coefficient $B$ is given by $B=16\pi^2T_{\text{reg}}(0)$ as before. In Section~\ref{sec:method}, we apply these ideas to stationary worldlines, resulting in formulae for the point-split energy density in terms of functions easily computed from the Lorentzian distance between two points on the curve and a tetrad that is adapted to it, in a manner we describe. Most of the required conditions on $T_{\text{reg}}$ follow directly from this analysis,
and the values $A$ and $B$ are identified in terms of Taylor coefficients of these functions. Appendix~\ref{sec:details} gives more detail on our methods, while in Appendix~\ref{sec:Taylor} the relevant Taylor coefficients are evaluated in terms of curvature invariants thus establishing~\eqref{eq:curvatureinvA} and~\eqref{eq:curvatureinvB}. In Section~\ref{sec:prototypes}, we work through each prototype in turn, providing explicit formulae for the point-split energy density that allow the 
remaining technical condition to be verified, and also as a check on our Taylor series calculations. In three cases, (inertial worldlines, linearly accelerated worldlines and the semicubical parabola), a closed form may be found for $\hat{T}$, and we can also check
our calculations by using~\eqref{eq:genQEI2} and~\eqref{eq:Qboundinitial}. Finally, 
in Section~\ref{sec:discussion}, we discuss the physical significance of our results and some open problems. Two further appendices contain additional computations: Appendix~\ref{sec:wicksqanalysis} computes a quantum inequality for the Wick square along stationary worldlines following the same general method of the main text, while Appendix~\ref{sec:Rindler} records the calculation of the minimally coupled stress-energy tensor in the Rindler vacuum and Rindler thermal states, which is needed for our discussion. 
 
\section{Stationary worldlines}\label{sec:stationary}

Throughout this paper we work on $4$-dimensional Minkowski spacetime, with metric $\eta=\diff{t}^2-\diff{x}^2-\diff{y}^2-\diff{z}^2$, and we employ the inertial coordinates $(t,x,y,z)$ except where otherwise specified. 
A stationary worldline is
any smooth curve $\gamma:\mathbb{R}\to \mathbb{R}^4$, whose velocity vector $\dot{\gamma}$ is a future-pointing unit timelike vector evolving under a $1$-parameter subgroup of the Lorentz group $SO(1,3)$, i.e., 
\begin{equation}
    \dot{\gamma}^\mu(s) = \exp(sM)^\mu_{\phantom{\mu}\nu}\dot{\gamma}^\nu(0),
\end{equation}
where $M$ is any fixed element of $\mathfrak{so}(1,3)$ (which requires
precisely that $M_{\mu\nu}$ is antisymmetric). As every component of $\exp(sM)$ is analytic in $s$, it follows that the Cartesian components of $\dot{\gamma}(s)$ and, integrating, the Cartesian coordinates of $\gamma(s)$, are also $s$-analytic. An equivalent definition of a stationary worldline is that $\gamma$ is an orbit of a future-pointing timelike Killing vector field 
\begin{equation}
    \xi^\mu(x) = M^\mu_{\phantom{\mu}\nu} (x^\nu- \gamma(0)^\nu)+ \dot{\gamma}^\mu(0),
\end{equation}
which is necessarily timelike in a neighbourhood of $\gamma$ and a future-pointing unit vector on $\gamma$.

Finally, stationary worldlines can also be described as the solutions to the Frenet-Serret equations with constant curvatures~\cite{Kottler1912,Synge1967,Letaw:1981}. 
Here, the curvature invariants of a general timelike curve $\gamma(s)$, parameterised by proper time, are defined as follows. Suppose a right-handed tetrad $e_a^\mu$ has been chosen along $\gamma$ so that 
\begin{equation}\label{eq:adapted}
\gamma^{(k+1)}(s)\in \Span\{e_0(s),\ldots,e_k(s)\} \qquad (0\le k\le 3),\qquad \text{and}\qquad
\dot{\gamma}(s)=e_0(s),
\end{equation}
in which case we say that $e_a^\mu$ is \emph{adapted} to $\gamma$. If the tetrad also satisfies 
\begin{equation}\label{eq:FStetrad} 
e_1(s)^\mu \ddot{\gamma}(s)_\mu\le 0,\qquad e_2(s)^\mu \dddot{\gamma}(s)_\mu\le 0,
\end{equation}
then it will be called a Frenet--Serret tetrad.  If 
the tetrad is defined by $e_a(s)=\exp(sM)e_a(0)$, then it is adapted (respectively, Frenet--Serret) if and only if~\eqref{eq:adapted} holds at $s=0$ (resp.,~\eqref{eq:adapted} and~\eqref{eq:FStetrad} hold at $s=0$). Explicit formulae resulting from a Gram--Schmidt procedure are given in~\cite{Letaw:1981}.
Expanding the derivatives of the tetrad vectors in terms of the tetrad, one obtains the generalized Frenet--Serret equations 
\begin{equation}\label{eq:FS}
    \dot{e}_a^\mu = K_a^{\phantom{a}b} e_b^\mu,
\end{equation}
where $K_{ab}$ is antisymmetric and tridiagonal (due to~\eqref{eq:adapted}). Thus it takes the form
\begin{equation}
    K_{\bullet\bullet}(s) = \begin{pmatrix} 
    0 & -\kappa(s) & 0 & 0 \\
    \kappa(s) & 0 & -\tau(s) & 0 \\
    0  & \tau(s) & 0 & -\upsilon(s) \\
    0 & 0 & \upsilon(s) & 0
    \end{pmatrix},
\end{equation}
which defines the curvature $\kappa$, torsion $\tau$ and hypertorsion $\upsilon$. Here, and elsewhere in this paper, bullets are used to indicate tensorial type, 
when displaying tensorial components in vector or matrix form. Explicitly, one has
\begin{equation}\label{eq:curveinvs}
    \kappa = e_{0\mu}\dot{e}_1^\mu = -e_{1\mu}\dot{e}_0^\mu, \qquad \tau = e_{1\mu}\dot{e}_2^\mu=-e_{2\mu}\dot{e}_1^\mu, \qquad \upsilon= e_{2\mu}\dot{e}_3^\mu=-e_{3\mu}\dot{e}_2^\mu.
\end{equation}
The choices made when specifying the Frenet--Serret tetrad ensure that $\kappa$ and $\tau$ are nonnegative, while $\upsilon$ can take any real value.

As the curvature invariants are constant along stationary worldlines, it is easy to compute higher derivatives of the tetrad, 
\begin{equation}\label{eq:tetradderivatives}
\frac{\diff^k}{\diff s^k}e_a^\mu = (K^k)_a^{\phantom{a}b} e_b^\mu, \qquad
(K^k)_a^{\phantom{a}b}= K_a^{\phantom{a}c_1}K_{c_1}^{\phantom{c_1}c_2}\cdots K_{c_{k-1}}{}^{b}.
\end{equation}
For example, the first three derivatives of the velocity $u=\dot{\gamma}$ may be computed as 
\begin{equation}\label{eq:e0_derivs}
     \dot{u}^\mu = \dot{e}_0^\mu = \kappa e_1^\mu,\qquad
     \ddot{u}^\mu = \kappa^2 e_0^\mu + \kappa\tau e_2^\mu,
     \qquad
     \dddot{u}^\mu = \kappa(\kappa^2-\tau^2)e_1^\mu + \kappa\tau\upsilon e_3^\mu.
\end{equation}

It is also possible to give a general formula for $\gamma(s)$ in terms of
$M$, $\gamma(0)$ and $\dot{\gamma}(0)$. As $M^\bullet_{\phantom{\bullet}\bullet}$ is antisymmetric with respect to $\eta$, there is a unique decomposition
\begin{equation}
    \dot{\gamma}(0)^\mu = M^\mu_{\phantom{\mu}\nu} v^\nu + k^\mu,
\end{equation}
where $M^\mu_{\phantom{\mu}\nu} k^\nu=0$. One then has
\begin{equation}
 \gamma(s)^\mu = \exp(sM)^\mu_{\phantom{\mu}\nu}v^\nu + s k^\mu + \gamma(0)^\mu - v^\mu.
\end{equation}

Any stationary worldline $\gamma$ may be related to one of six basic types by a proper orthochronous Poincar\'e transformation. Note that $\gamma(s)$ is determined by the initial position, $\gamma(0)\in \mathbb{R}^4$, the initial four-velocity $\dot{\gamma}(0)$ and the 
element $M\in\mathfrak{so}(1,3)$ that fixes the evolution $\dot{\gamma}(s) = \exp(sM)\dot{\gamma}(0)$. Under a Poincar\'e transformation $x\mapsto \Lambda x+w$, $\gamma$ is mapped to $\tilde{\gamma}(s) =\Lambda \gamma(s) + w$, whose velocity
evolves according to the $1$-parameter Lorentz subgroup $\exp(s\Lambda M\Lambda^{-1})$ and which is therefore also a stationary worldline. As the Lorentz transformation maps a Frenet--Serret tetrad for $\gamma$ to a Frenet--Serret tetrad for $\tilde{\gamma}$, it follows from~\eqref{eq:curveinvs} that the curvature invariants of $\tilde{\gamma}$ are identical to those of $\gamma$. Using the classification of conjugacy classes in $\mathfrak{so}(1,3)$~\cite{Shaw:1970}, we may choose $\Lambda$ in such a way that
$\tilde{M}=\Lambda M\Lambda^{-1}$ is one of five possible types:
(a) the zero element, generating the trivial subgroup of $SO(1,3)$, (b) a generator of boosts in the $tx$-plane, corresponding to a hyperbolic subgroup of $SO(1,3)$, (c) a generator of rotations in the $yz$-plane, corresponding to an elliptic subgroup of $SO(1,3)$, (d) a generator of a null rotation that fixes the null vector $\partial_t+\partial_x$ but acts nontrivially on all other null vectors, corresponding to a parabolic subgroup of $SO(1,3)$; (e) the sum of a generator of boosts in the $tx$-plane and a generator of rotations in the $yz$ plane, corresponding to a loxodromic subgroup of $SO(1,3)$. In each case, Lorentz transformations that commute with the $1$-parameter subgroup in question can be used to arrange that $\dot{\tilde{\gamma}}(0)$ takes a convenient form. 

Taking these possibilities in turn: in case (a), all Lorentz transformations commute with the trivial subgroup, so we may without loss assume that $\tilde{\gamma}(s)=(s,0,0,0)$. In case (b), the subgroup of boosts parallel to the $x$-axis commutes with itself and the subgroup of rotations in the $yz$-plane. Thus, we may arrange that $\dot{\tilde{\gamma}}(0) = \cosh\chi\partial_t + \sinh\chi \partial_y$ for some $\chi\in\mathbb{R}$,\footnote{We could even arrange that $\chi\ge 0$, but it is convenient not to insist on this.} leading to two subcases: $\chi=0$, in which case (after possible translation) 
\begin{equation}
\tilde{\gamma}(s)=(a^{-1}\sinh as, a^{-1}\cosh as,0,0)
\end{equation}
is a uniformly linearly accelerated worldline with $a\neq 0$, or $\chi\neq 0$, in which case (up to translations)
\begin{equation}
\tilde{\gamma}(s) = (a^{-1}\cosh\chi\,\sinh as,a^{-1}\cosh\chi\,\cosh as,-s\sinh\chi,0)
\end{equation}
is a catenary. The curvature invariants (in either subcase) are $\kappa =|a|\cosh\chi$ and $\tau=|a\sinh\chi|$, while the hypertorsion is $\upsilon=0$. For convenience, the curvature invariants for all six prototypes are tabulated in Table~\ref{tab:inv}, in agreement with~\cite{LetawPfautsch:1982}.

\begin{table}
    \centering
    \begin{tabular}{c|cccccc}
         & Inertial & Linear Acc. &  Catenary &  Parabolic & Elliptic  & Loxodromic  \\ 
         & $\kappa=\tau=\upsilon=0$ & $\kappa>0$ & $\kappa>\tau>0$ & $\kappa=\tau>0$ & $\tau>\kappa>0$ & $\kappa,\tau>0$\\
         & &  $\tau=\upsilon=0$ & $\upsilon=0$ & $\upsilon=0$ & $\upsilon=0$ & $\upsilon\neq 0$
         \\ \hline
    $\kappa$ & $0$      & $\lvert a \rvert$  & $\lvert a\rvert \cosh\chi$ &$\lvert a \rvert$     & $r\omega^2$ &  $\sqrt{C^2 a^2+V^2\omega^2}$ \\
    $\tau$   & $0$      & $0$  & $\lvert a \sinh\chi\rvert$   & $\lvert a\rvert$       & $\lvert\omega\rvert\sqrt{1+(r\omega)^2}$  & $(a^2+\omega^2)C|V|/\kappa$ \\
    $\upsilon$ & $0$    & $0$ & $0$ & $0$ & $0$ & $a \omega/\kappa $
    \end{tabular}
    \caption{Curvature invariants for the stationary worldlines.} 
    \label{tab:inv}
\end{table}

In case (c), the $1$-parameter parabolic subgroup takes the form
\begin{equation}\label{eq:P}
    P^\bullet_{\phantom{\bullet}\bullet}(s) = \begin{pmatrix}
    1+(as)^2/2 & -(as)^2/2 & 0 & as \\ 
    (as)^2/2 & 1-(as)^2/2 & 0 & as  \\
    0 & 0 & 1 & 0 \\
    as & -as & 0 & 1
    \end{pmatrix}=\exp \begin{pmatrix}
    0 & 0 & 0 & as \\ 
   0 & 0 & 0 & as  \\
    0 & 0 & 0 & 0 \\
    as & -as & 0 & 0
    \end{pmatrix}
\end{equation}
for some constant nonzero $a\in\mathbb{R}$, and commutes with Lorentz transformations of the form
\begin{equation}
    \Lambda^\bullet_{\phantom{\bullet}\bullet}=
    \begin{pmatrix}
    1+r^2/2 & -r^2/2 & r\cos\theta & r\sin\theta \\ 
    r^2/2 & 1-r^2/2 & r\cos\theta & r\sin\theta \\
    r\cos\theta & -r\cos\theta & 1 & 0 \\
    r\sin\theta & -r\sin\theta & 0 & 1
    \end{pmatrix}
\end{equation}
which can be used to bring the initial velocity into
the form $\dot{\tilde{\gamma}}(0) = \cosh\chi \partial_t + \sinh\chi\partial_x$ for some $\chi\in\mathbb{R}$. Conjugating $P^\bullet_{\phantom{\bullet}\bullet}(s)$ with a
boost in the $tx$-plane results in $P^\bullet_{\phantom{\bullet}\bullet}(\lambda s)$ for some $\lambda>0$; 
in other words effectively rescaling $a$. Therefore there is no loss
of generality in assuming that the initial $4$-velocity is $\dot{\tilde{\gamma}}(0)=\partial_t$, in which case the worldline (up to translation) is the  semicubical parabola, 
\begin{equation} \label{eq:semipara}
    \tilde{\gamma}(s) = \left(s+\dfrac{1}{6}a^2s^3, \dfrac{1}{6}a^2s^3, 0, \dfrac{1}{2}as^2\right).
\end{equation}

Next, the elliptic subgroup in case (d) commutes with boosts in the $tx$-plane and rotations in the $yz$-plane. Accordingly, we may
arrange the initial velocity to be $\dot{\tilde{\gamma}}(0) = \cosh\chi\partial_t + \sinh\chi \partial_z$ for some $\chi\in\mathbb{R}$; the special case $\chi=0$ corresponds to inertial motion and may be discarded. 
Up to a translation, this results in the uniformly rotating worldline
\begin{equation}
    \tilde{\gamma}^\bullet(s) = \left(s\cosh\chi, 0, r\cos{\omega s}, r\sin\omega s\right),
\end{equation}
where the radius $r>0$ and proper angular velocity $\omega\neq 0$ are related to the initial rapidity by $r\omega = \sinh\chi$. The proper acceleration is $\kappa = r\omega^2$, while the torsion is $\tau=\lvert\omega\rvert\sqrt{1+(r\omega)^2}$ and the hypertorsion vanishes. 

Lastly, in case (e), the loxodromic subgroup is generated by a linear
combination of a $tx$-boost generator and a $yz$-rotation generator.
As it commutes with $tx$-boosts and $yz$-rotations, we may assume without loss that the initial velocity is $\dot{\tilde{\gamma}}(0) = 
\cosh\chi\partial_t + \sinh\chi\partial_z$ for $\chi\in\mathbb{R}\setminus\{0\}$; the possibility $\chi=0$ corresponds to a hyperbolic worldline and is rejected. Up to a translation, this results in the worldline
\begin{equation} \label{eq:loxotraj}
    \gamma^\bullet(s) = (Ca^{-1}\sinh(a s),Ca^{-1}\cosh(a s),V\omega^{-1}\cos(\omega s),V\omega^{-1}\sin(\omega s)),
\end{equation} 
where $C=\cosh \chi$ and $V=\sinh\chi$, 
which undergoes both rotation in the $yz$-plane at constant proper angular velocity $\omega\neq 0$ and constant distance $|V/\omega|$ from the $x$-axis, while undergoing uniform acceleration in the $x$-direction controlled by $a\neq 0$ (the cases where one or both of $a$ or $\omega$ vanish are already covered under~(a), (b) and (d)). The curvature invariants for this worldline are
\begin{equation}\label{eq:loxoinv}
    \kappa = \sqrt{C^2a^2+V^2\omega^2}, \qquad
    \tau= (a^2+\omega^2)C|V|/\kappa, \qquad
    \upsilon=a \omega/\kappa .
\end{equation} 

\section{Reformulation of the QEI bound} \label{sec:genQEI} 

We study the massless minimally coupled scalar field in $4$-dimensional Minkowski spacetime, with field equation $\Box \phi=\eta^{\mu\nu}\nabla_\mu\nabla_\nu\phi=0$ and classical stress-energy tensor
\begin{equation}
T_{\mu\nu} = (\nabla_\mu\phi)\nabla_\nu\phi - 
\tfrac{1}{2}\eta_{\mu\nu}\eta^{\alpha\beta}
(\nabla_\alpha\phi)\nabla_\beta\phi.
\end{equation}
Consider an observer following a timelike curve $\gamma$, parameterised by proper time, with $4$-velocity $u^\mu=\dot{\gamma}^\mu$. This
observer sees energy density 
\begin{equation}\label{eq:classicalenergydensity}
    T_{\mu\nu}u^\mu u^\nu = \frac{1}{2}\sum_{a=0}^3 
(e_a^\mu \nabla_\mu\phi)^2,
\end{equation}
where $e_a^\mu$ ($0\le a\le 3$) is a tetrad defined around $\gamma$
with $e_0^\mu|_\gamma=u^\mu$.

In quantum field theory, the stress-energy tensor requires renormalisation. Let 
\begin{equation}
G(x,x')=\langle \phi(x)\phi(x')\rangle_\omega
\end{equation}
be the Wightman function of the field in a state $\omega$. The Wick square has expectation value
\begin{equation}
    \langle {:}\phi^2(x){:}\rangle_\omega = (G-G_0)(x,x) ,
\end{equation}
where 
\begin{equation}
G_0(x,x') = \lim_{\epsilon\to 0+} \frac{-1}{4\pi^2(
(t-t'-i\epsilon)^2-\|\xb-\xb'\|^2)}
\end{equation}
is the Wightman function of the Poincar\'e invariant vacuum $\omega_0$. This expression makes sense if (like $\omega_0$) $\omega$ is a \emph{Hadamard state}~\cite{KAY1991,Moretti:2021}, because the difference $G-G_0$ is then a smooth function.
Similarly, the renormalised stress-energy tensor has expectation value
\begin{equation}
    \langle {:}T_{\mu\nu}(x){:}\rangle_\omega = 
    D_{\mu\nu}(x) - \tfrac{1}{2}\eta_{\mu\nu}\eta^{\alpha\beta}
    D_{\alpha\beta}(x),
\end{equation}
where
\begin{equation}
    D_{\mu\nu}(x)=\coin{(\nabla\otimes\nabla)(G-G_0)}_{\mu\nu}(x)
\end{equation}
and the double square brackets denote a coincidence limit.

Although the classical energy density~\eqref{eq:classicalenergydensity} is everywhere nonnegative, the quantised energy density may assume negative expectation values. The QEIs provide lower bounds on averaged expectation values, for which a prototype is a lower bound on the following expression
\begin{equation} \label{eq:QEIorig}
    \int \diff{s} |g(s)|^2\langle\mathbf{:}(\mathcal{Q}\phi)^2\mathbf{:}\rangle_\omega(\gamma(s)),
\end{equation}
where $\mathcal{Q}$ is a partial differential operator with smooth real coefficients and $g \in \mathcal{C}_0^\infty(\mathbb{R})$ is a smooth
real-valued test function. In the case where $\mathcal{Q}$ is the identity,~\eqref{eq:QEIorig} is an averaged Wick square, while by considering a sum of similar terms for $\mathcal{Q}_a=2^{-1/2} e^\mu_a\nabla_\mu$ for $0\le a\le 3$, we can bound averages of the energy density along $\gamma$.  

A lower bound on~\eqref{eq:QEIorig} was established in~\cite{GenQEI2000CFewster} -- in fact the bound applies to general timelike curves in arbitrary globally hyperbolic spacetimes for massive as well as massless fields. In our case it asserts that
\begin{equation}\label{eq:urQEI}
    \int_{-\infty}^\infty \diff{s}   |g(s)|^2\langle\mathbf{:}(\mathcal{Q}\phi)^2\mathbf{:}\rangle_\omega(\gamma(s))\ge
    -\int_0^\infty \frac{\diff{\alpha}}{\pi} \yhwidehat{g\otimes g T}(-\alpha,\alpha) >-\infty
\end{equation}
holds for all real-valued compactly supported smooth test functions $g$, 
and all Hadamard states $\omega$, where 
\begin{equation}
    T(s,s') = \langle \mathcal{Q}\phi(\gamma(s))\mathcal{Q}\phi(\gamma(s'))\rangle_{\omega_0}=((\mathcal{Q}\otimes \mathcal{Q})G_0)(\gamma(s),\gamma(s')).
\end{equation}
Here, the vacuum two-point function enters because normal ordering is performed relative to the vacuum; the general results of~\cite{GenQEI2000CFewster} also allow for any Hadamard state to be used as the reference state for this purpose. At a more formal level, 
$T$ is the pull-back of the distribution $(\mathcal{Q}\otimes \mathcal{Q})G_0$ by the map $(s,s')\mapsto (\gamma(s),\gamma(s'))$, and its existence is owed to the special properties of the Hadamard condition and the fact that $\gamma$ is timelike -- see~\cite{GenQEI2000CFewster} for full details and rigorous proofs. 

As already mentioned, a QEI for the energy density involves a sum of such bounds, leading to~\eqref{eq:genQEI} with 
\begin{equation}
    T(s,s') = \frac{1}{2}\sum_{a=0}^3 ((\nabla_{e_a}\otimes\nabla_{e_a})G_0)(\gamma(s),\gamma(s')).
\end{equation} 
While it is usually not hard to obtain the distribution $T$ for a given timelike curve in Minkowski spacetime, assuming that $G_0$ is given, it is not usually possible to find the Fourier transform required to compute the QEI bound~\eqref{eq:urQEI} in closed form. 

The situation is somewhat simplified if $T(s,s')$ is translationally invariant, in which case one has
the bound given by~\eqref{eq:genQEI2} and~\eqref{eq:Qboundinitial}. This can be taken a little further, on observing that $|\hat{g}(\alpha)|^2$ is even, so only the even part $Q_\text{even}(\alpha)=\tfrac{1}{2}(Q(\alpha)+Q(-\alpha))$ of $Q$ contributes to~\eqref{eq:genQEI2}, resulting in the bound 
\begin{equation}\label{eq:QEIQeven}
    \int \diff{s} |g(s)|^2\langle\mathbf{:}(\mathcal{Q}\phi)^2\mathbf{:}\rangle_\omega(\gamma(s)) \geq - \int_{-\infty}^\infty \diff{\alpha} \lvert \hat{g}(\alpha) \rvert^2 Q_{\text{even}}(\alpha),
\end{equation}
which is the final form of our prototypical quantum inequality.

A convenient expression for $Q_\text{even}$ may be found by manipulating equation \eqref{eq:Qboundinitial} in the following way:
\begin{align}\label{eq:Qeven}
    Q_{\text{even}}(\alpha) 
    &= \dfrac{1}{4\pi^2}\left[\int_{-\infty}^\alpha \hat{T}(u)\diff{u}+\int_{-\infty}^{-\alpha} \hat{T}(u)\diff{u}\right] \nonumber \\
    &=  \dfrac{1}{4\pi^2}\left[2\int_{-\infty}^0 \hat{T}(u)\diff{u}+\int_{0}^\alpha \hat{T}(u)\diff{u} -\int_0^{\alpha} \hat{T}(-u)\diff{u}\right] \nonumber \\
    &= \dfrac{1}{2\pi^2}\left[\int_{-\infty}^0 \hat{T}(u)\diff{u}+\int_{0}^\alpha \hat{T}_{\text{odd}}(u)\diff{u}\right],
\end{align}
where $\hat{T}_{\text{odd}}(u) = \tfrac{1}{2}(\hat{T}(u)-\hat{T}(-u))$. In the above calculation, $\hat{T}$ is assumed to be continuous, as is the case for the examples we will study.

Evaluating $Q_{\text{even}}$ from~\eqref{eq:Qeven} requires several steps.  Computing $T$ is a tedious but straightforward calculation best handled using computer algebra. In the simplest cases, the transform may be evaluated in closed form, which (as will be seen later) is the case for the inertial, uniformly accelerated and semicubical parabola worldlines, but is not possible (to our knowledge) in the case of the other stationary worldlines. However, this obstacle can be circumvented, as we now describe.

Using the Minkowski vacuum as the reference state, we will show in Section~\ref{sec:method} that the point-split energy density along a stationary worldline may be written in the form
\begin{equation}\label{eq:Tsingreg}
    T(s,s') = T_{\text{sing}}(s-s') + T_{\text{reg}}(s-s'),
\end{equation}
where $T_{\text{sing}}$ is given by the distributional limit
\begin{equation}\label{eq:Tsing}
    T_{\text{sing}}(s) =\lim_{\epsilon\to 0+}\left( \frac{3}{2\pi^2(s-i\epsilon)^4} -\frac{A}{4\pi^2 (s-i\epsilon)^2}\right)
\end{equation}
for some constant $A$ (the sign is chosen for later convenience) and $T_{\text{reg}}$ is smooth, real and even, and decaying as $\mathcal{O}(s^{-2})$ as $|s|\to\infty$.  In particular, $T_{\text{reg}}$ is absolutely integrable and has a well-defined Fourier transform that is continuous, real and even. Therefore it does not contribute to $\hat{T}_{\text{odd}}$. Turning to $T_{\text{sing}}$, its leading singularity is universal, essentially because all stationary worldlines resemble inertial worldlines on sufficiently short timescales. The specific coefficient is fixed by the Hadamard form and the definition of the energy density along the curve. Meanwhile the coefficient $A$ carries information about the specific curve at hand. The Fourier transform of $T_{\text{sing}}$, in our convention, is
\begin{equation}
    \hat{T}_{\text{sing}}(u) = \dfrac{1}{2\pi} (u^3+Au) \Theta(u),
\end{equation}
where $\Theta$ is the Heaviside distribution.
Evidently $T_{\text{sing}}$ does not contribute to the first term in~\eqref{eq:Qeven}, while the odd part of $\hat{T}$ is
\begin{equation}
    \hat{T}_{\text{odd}}(u) = \dfrac{1}{4\pi} (u^3+Au) ,
\end{equation}
recalling that $\hat{T}_{\text{reg}}$ is even. We now have $Q_{\text{even}}$ in the form
\begin{align}\label{eq:Qevenagain}
    Q_{\text{even}}(\alpha) &= \dfrac{1}{2\pi^2} \left[ \int_{-\infty}^0 \diff{u} \,\hat{T}_{\text{reg}}(u)  + \dfrac{1}{4\pi}\int_0^\alpha  \diff{u} (u^3+Au) \right] \notag \\
    &= \dfrac{1}{32\pi^3}(\alpha^4 + 2 A \alpha^2) + \dfrac{T_{\text{reg}}(0)}{2\pi},
\end{align} 
where we have again used the evenness of $\hat{T}_{\text{reg}}$ and the Fourier inversion formula.
Inserting~\eqref{eq:Qevenagain} into~\eqref{eq:QEIQeven} and using Parseval's theorem gives the QEI bound
\begin{equation} \label{eq:QEIfinalgeneral}
    \int \diff{s} |g(s)|^2\langle\mathbf{:}T_{\mu\nu} \dot{\gamma}^\mu\dot{\gamma}^\nu\mathbf{:} \rangle_\omega(\gamma(s)) \geq -\dfrac{1}{16\pi^2}\int_{-\infty}^\infty \diff{s} \left(\lvert g''(s) \rvert^2 + 2A \lvert g'(s) \rvert^2 +  B \lvert g(s) \rvert^2\right),
\end{equation}
where $B=16\pi^2T_{\text{reg}}(0)$.

The upshot of this analysis is a direct route to the QEI once the point-split expression $T$ is obtained; all that is needed is to isolate the appropriate values of $A$ and $T_{\text{reg}}(0)$, avoiding the need to compute $\hat{T}$ explicitly. This apparent royal road is made possible because of the special structure of the Minkowski vacuum two-point function for the massless scalar field in four dimensions -- closely related to Huygens' principle. A similar analysis for a QI on the Wick square can be found in Appendix \ref{sec:wicksqanalysis}.

\section{Computation of the point-split energy density}\label{sec:method}

In this section we establish that the point-split energy density
along stationary worldlines obeys equations~\eqref{eq:Tsingreg} 
and~\eqref{eq:Tsing}, and also that $T_{\text{sing}}$ and $T_{\text{reg}}$ have the properties mentioned above, with one exception that will be treated by examining the six prototypical cases in Section~\ref{sec:prototypes}.

Let $\gamma$ be any stationary worldline with $\dot{\gamma}(s) = \exp(sM)\dot{\gamma}(0)$ and $\dot{\gamma}(0)$ a future-pointing unit timelike vector.
Suppose that 
\begin{equation}\label{eq:eadef}
    e_a(s) = \exp (sM) e_a(0) \qquad (0\le a\le 3)
\end{equation}  
is an adapted frame on $\gamma$ satisfying~\eqref{eq:adapted}. In general there may be many possible adapted tetrads of this type. However, if $\tilde{e}_a(s)$ is any other then it is related to $e_a(s)$ by a rigid rotation, i.e., $\tilde{e}_0(s)=e_0(0)$ and $\tilde{e}_i(0)=R_i^{\phantom{i}j}e_j(0)$ (summing $j$ over $1,2,3$), where $\delta^{im}R_i^{\phantom{i}j}R_m^{\phantom{m}n}=\delta^{mn}$, $\det R=1$. This must be true for some $R$ at $s=0$, and extends to all $s$ as both tetrads evolve under $\exp(sM)$.

Next, recall that the vacuum $2$-point function may be given as a distributional limit
\begin{equation} \label{eq:2pointfndistributionallimit}
    G_0(x,x') = \lim_{\epsilon\to 0+} F(\sigma_\epsilon(x,x'))
\end{equation}
where $F(z) = 1/(4\pi^2 z)$ and 
\begin{equation}
    \sigma_\epsilon(x,x') = -\eta_{\mu\nu} (x-x'-i\epsilon \partial_t)^\mu
    (x-x'-i\epsilon \partial_t)^\nu
\end{equation} 
is the regulated signed squared geodesic separation of $x$ and $x'$. As usual, we have identified Minkowski spacetime with its tangent spaces at all points.

Distributional derivatives may be taken under the limit in~\eqref{eq:2pointfndistributionallimit}, giving
\begin{equation}
   \tfrac{1}{2} (\nabla_\mu\otimes 1)G_0(x,x') = -\lim_{\epsilon\to 0+} F'(\sigma_\epsilon(x,x')) (x-x'-i\epsilon \partial_t)_\mu
\end{equation}
and
\begin{equation}
    \tfrac{1}{2}(\nabla_\mu\otimes \nabla_\nu)G_0(x,x') = \lim_{\epsilon\to 0+}
    \left(F'(\sigma_\epsilon(x,x')) \eta_{\mu\nu} - 2F''(\sigma_\epsilon(x,x'))
    (x-x'-i\epsilon \partial_t)_\mu (x-x'-i\epsilon \partial_t)_\nu\right).
\end{equation}  
Contracting with $e_a(x)^\mu e_a(x')^\nu$ (without summing on $a$) and pulling back to the worldline, we find
\begin{align}\label{eq:diff2G0}
   \tfrac{1}{2} ((\nabla_{e_a}\otimes  \nabla_{e_a})G_0) (\gamma(s),\gamma(s')) &= 
    \lim_{\epsilon\to 0+}
    F'(\sigma_\epsilon(\gamma(s),\gamma(s'))) C_a(s,s')   \notag\\
    &\qquad\qquad 
     +\lim_{\epsilon\to 0+} 2 F''(\sigma_\epsilon(\gamma(s),\gamma(s')))
    D_a(s,s') D_a(s',s) .
\end{align}
(note the order of variables in the last two factors in the second term) where 
\begin{equation} \label{eq:CandDdefn}
    C_a(s,s') = \eta_{\mu\nu}e_a^\mu(s)e_a^\nu(s') , \qquad D_{a}(s,s') = 
    (\gamma(s)-\gamma(s'))_\mu e_a^\mu(s).
\end{equation}
Under a change of frame from $e_a$ to $\tilde{e}_a$ as described above, one has $\tilde{C}_0=C_0$, $\tilde{D}_0=D_0$, while $\tilde{D}_i=R_i^{\phantom{i}j} D_j$ and 
$\tilde{C}_i(s,s')=R_i^{\phantom{i}j}R_i^{\phantom{i}k}\eta_{\mu\nu}e_j^\mu(s)e_k^\nu(s')$. By
orthogonality, this implies that $\sum_{i=1}^3 \tilde{C}_i(s,s')=\sum_{i=1}^3 C_i(s,s')$
and $\sum_{i=1}^3 \tilde{D}_a(s,s')\tilde{D}_a(s',s)=\sum_{i=1}^3 D_a(s,s')D_a(s',s)$.

In Appendix~\ref{sec:details}, we give some further details to justify the above distributional manipulations
and prove the following result, where $\kappa$, $\tau$ and $\upsilon$ are the curvature invariants of $\gamma$ as described in Section~\ref{sec:stationary}. 
\begin{lemma}
(a) With the choice of tetrad just described, $C_a(s,s')$ and $D_a(s,s')$ are translationally invariant,
depending only on $s-s'$. There are entire analytic functions $G_a$ and $H_a$ such that
\begin{equation} \label{eq:CtoGandDtoH}
C_a(s,s')=G_a(\kappa^2(s-s')^2), \qquad D_a(s,s')D_a(s',s) = -(s-s')^2 H_a(\kappa^2(s-s')^2),
\end{equation}
where, in the limit $z\to 0$,
\begin{equation}\label{eq:sumGa}
    \sum_{a=0}^3 G_a(z)  = -2 + \frac{\tau^2+\upsilon^2}{\kappa^2}z + 
\frac{(\kappa\tau)^2- (\tau^2+\upsilon^2)^2}{\kappa^4} z^2 + \mathcal{O}(z^3),
\end{equation}
and
\begin{equation}\label{eq:sumHa}
\sum_{a=0}^3 H_a(z) = 1 + \frac{z}{12} + \frac{\kappa^2+19\tau^2}{360\kappa^2}z^2 +\mathcal{O}(z^3).
\end{equation} 

(b) The signed square geodesic separation of points along $\gamma$ obeys
\begin{equation}
    \sigma_0(\gamma(s),\gamma(s')) = -(s-s')^2 \Upsilon(\kappa^2(s-s')^2),
\end{equation}
where $\Upsilon$ is entire analytic with 
\begin{equation} \label{eq:upsilongeneral} 
\Upsilon(z) = 1 + \frac{1}{12}z +
 \frac{\kappa^2-\tau^2}{360\kappa^2} z^2 + \mathcal{O}(z^3) 
\end{equation} 
as $z\to 0$. Furthermore, for  $z\in [0,\infty)$, $\Upsilon(z)$ is real with $\Upsilon(z)\ge 1$.
\end{lemma}

The Lemma now allows us to compute the point-split energy density by evaluating the right-hand side of~\eqref{eq:diff2G0} and summing over $a$. 
We use the fact (explained in Appendix~\ref{sec:details}) that
\begin{equation}\label{eq:keyfact}
    \lim_{\epsilon\to 0+} \frac{(s-s')^{2j}}{\sigma_\epsilon(\gamma(s),\gamma(s'))^k} = 
    \frac{(-1)^k}{\Upsilon(\kappa^2(s-s')^2)^k}
    \lim_{\epsilon\to 0+} \frac{1}{(s-s'-i\epsilon)^{2(k-j)}},
\end{equation}
where the limits are taken in the sense of distributions, as is the multiplication by a smooth prefactor on the right-hand side. If $j=k$, the distributional limit on the right-hand side may be replaced by unity. 
In particular, when calculating $T(s,s')$ from~\eqref{eq:diff2G0}, the factor $(s-s')^2$ in $D_a(s,s')D_a(s',s)$ cancels a factor of $(s-s'-i\epsilon)^2$ in the denominator, as $\epsilon\to 0+$.
The upshot is that 
\begin{equation}\label{eq:K}
    T(s,s') =-\frac{1}{4\pi^2}
    \lim_{\epsilon\to 0+} \frac{K(\kappa^2(s-s')^2)}{(s-s'-i\epsilon)^4}, \qquad\text{where}\qquad
    K(z) = 
    \sum_{a=0}^3 \left(
\frac{ G_a(z)}{\Upsilon(z)^2} - 4 \frac{ H_a(z)}{\Upsilon(z)^3} \right)
\end{equation}
is a meromorphic function that is analytic in a neighbourhood of the positive real axis (on which 
$\Upsilon$ is bounded away from zero).

The singular part is easily isolated by splitting off the first two terms of the Taylor series for $K$ from the remainder, which carries a leading factor of $(s-s')^4$ that
cancels the denominator in the limit $\epsilon\to 0+$. Similarly, the $\mathcal{O}(z)$ part of the Taylor series partly cancels the denominator. Thus,
$T(s,s')=T_{\text{sing}}(s-s')+ T_{\text{reg}}(s-s')$ with
\begin{equation}
    T_{\text{sing}}(s) = -\frac{1}{4\pi^2}
    \lim_{\epsilon\to 0+} \frac{K(0)}{(s-i\epsilon)^4} - \frac{1}{4\pi^2}
    \lim_{\epsilon\to 0+} \frac{\kappa^2 K'(0)}{(s-i\epsilon)^2},
\end{equation}
and
\begin{equation}
    T_{\text{reg}}(s) = -\frac{\kappa^4}{4\pi^2}J((\kappa s)^2), \qquad\text{where}\qquad
    J(z)=\frac{K(z)-K(0)-K'(0)z}{z^2}
\end{equation} 
is analytic on a neighbourhood of the positive real axis, so $J((\kappa s)^2)$ is smooth for $s\in \mathbb{R}$.

Using the Lemma, we may read off that $K(0)=-6$, thus establishing \eqref{eq:Tsing}, with $A=\kappa^2 K'(0)$. 
Meanwhile, $T_{\text{reg}}(s)$ is smooth, even, and real-valued for $s\in\mathbb{R}$. 
Provided that $K(z)=\mathcal{O}(z)$ as $z\to\infty$ on the real axis, we find that $T_{\text{reg}}(s)=\mathcal{O}(s^{-2})$ as $s\to\infty$, which completes the properties needed in Section~\ref{sec:genQEI}. Furthermore,
\begin{equation}
     T_{\text{reg}}(0)=-\frac{J(0)\kappa^4}{4\pi^2}=-\frac{K''(0)\kappa^4}{8\pi^2}. 
\end{equation}
Note that if we had used the tetrad $\tilde{e}$ instead, the function $K$ would be unchanged, owing to the remarks before the Lemma. Thus the QEIs obtained from $\tilde{e}_a$ and $e_a$ are identical. 

These results now provide a calculational method to determine the QEI along stationary worldlines.
Starting from the generator $M\in\mathfrak{so}(1,3)$ and the initial $4$-velocity $u(0)$, choose a tetrad as described at the start of this section, and compute the proper acceleration $\kappa = \sqrt{-\eta(Mu(0),M u(0))}$. The translational invariance of $C_a$ and $D_a$ means that they can be calculated conveniently as
\begin{equation} \label{eq:CandDaltdefn}
      C_a(s,s') = \eta_{\mu\nu}e_a^\mu(s-s')e_a^\nu(0) , \qquad D_{a}(s,s') = 
    -(\gamma(s'-s)-\gamma(0))_\mu e_a^\mu(0),
\end{equation}
from which $G_a$ and $H_a$ are easily obtained. The function $\Upsilon$ is computed directly from 
the Lorentz interval between $\gamma(0)$ and $\gamma(s)$. Then construct $K(z)$ according to 
\eqref{eq:K} and check that $K(z)=\mathcal{O}(z)$ as $z\to\infty$. Then the QEI along $\gamma$ 
is given by~\eqref{eq:QEIfinalgeneral}, with constants  
\begin{equation}\label{eq:ABfromK}
    A=\kappa^2 K'(0),\qquad
    B=-2 \kappa^4 K''(0). 
\end{equation}

The constants $A$ and $B$ can be computed from the first few terms of the Taylor expansions of $\sum_a G_a$, $\sum_a H_a$ and $\Upsilon$, given in~\eqref{eq:sumGa},~\eqref{eq:sumHa} and~\eqref{eq:upsilongeneral} respectively. After a calculation, one finds
\begin{equation}
K(z)= -6 + z\dfrac{\kappa^2+\tau^2+\upsilon^2}{\kappa^2}-z^2\dfrac{1}{360\kappa^4}\left(3\kappa^4+62\kappa^2\tau^2+30(\kappa^2+\tau^2+\upsilon^2)^2\right) + \mathcal{O}(z^3),
\end{equation}
from which the formulae~\eqref{eq:curvatureinvA} and~\eqref{eq:curvatureinvB} follow immediately.
Nonetheless, this is perhaps not the most illuminating calculation and also does not provide
a check that $K(z)=O(z)$ for large real $z$, which was assumed above. For these reasons,
and their own intrinsic interest, we will also provide explicit calculations in Section~\ref{sec:prototypes} that together cover all possible stationary worldlines.

\section{QEIs for the prototypical stationary worldlines}\label{sec:prototypes}
 
We have now established the general QEI for stationary worldlines in Minkowski spacetime, 
assuming a technical condition on the growth of $K$. In this section,
we reduce the problem of computing the QEI for a general stationary worldline to six prototypical cases, which will be treated in turn. These calculations follow the method
of Section~\ref{sec:method} and result in explicit formulae for $K$. In this way it is seen that
the growth condition holds in all cases and we also obtain a check on the Taylor series calculations in~Appendix~\ref{sec:Taylor}. 

We have already discussed the fact that any stationary worldline 
may be brought into one of the six standard forms by a Poincar\'e transformation, without changing the curvature invariants. Owing to Poincar\'e invariance of the vacuum state, and because Poincar\'e invariance maps an adapted tetrad of the form $e_a(s)=\exp(sM)e_a(0)$
along the original curve to a tetrad with the same properties on the new one, the point-split energy density obtained by the method of Section~\ref{sec:method} is exactly the same for the two worldlines, which accordingly share the same QEI bound. 

The QEIs for the prototypical stationary worldlines are now given in turn. Most of the computations that follow were conducted using the computer algebra system Maple.
  
\subsection{Trivial subgroup: inertial motion} \label{sec:inertial}

For the inertial worldline $\gamma(s) = (s,0,0,0)$, we employ the adapted tetrad $\partial_t,\partial_x,\partial_y,\partial_z$, which is constant along $\gamma$, leading immediately to the relations $C_0(s,s')=1$, $C_i(s,s')=-1$ for $i=1,2,3$, while $D_0(s,s') = s-s'$, $D_i(s,s')=0$ for all $s,s'$. It follows that $G_0=H_0\equiv 1$, $G_i\equiv -1$, $H_i\equiv 0$. Furthermore, $\Upsilon\equiv 1$ because $\sigma_0(\gamma(s),\gamma(s'))=-(s-s')^2$. Hence $K\equiv -6$ and 
one finds $T(s,s')=T_{\text{sing}}(s-s')$ where
\begin{equation}
    T_{\text{sing}}(s) = \lim_{\epsilon\to 0+} \frac{3}{2\pi^2 (s-i\epsilon)^4}.
\end{equation}
Consequently $T_{\text{reg}}$ vanishes identically, and we may read off immediately that $A=B=0$, 
reproducing QEI~\eqref{eq:genQEI3} by substituting into~\eqref{eq:QEIfinalgeneral},
and in agreement with~\eqref{eq:curvatureinvA} and~\eqref{eq:curvatureinvB}.
Of course these results are easily obtained by direct differentiation of the two-point function; our purpose here is to show how they follow from formulae in Section~\ref{sec:method}. 

Alternatively, we may proceed by taking the Fourier transform
\begin{equation}\label{eq:That_inertial}
\hat{T}_{\text{sing}}(u) =  u^3 \Theta(u)/(2\pi),
\end{equation}
from which we obtain $Q(\alpha)= \alpha^4\Theta(\alpha)/(16\pi^3)$ by~\eqref{eq:Qboundinitial}, leading to~\eqref{eq:genQEI3} as discussed in the introduction. 

\subsection{Hyperbolic subgroups: linear acceleration} \label{sec:linaccel}

We consider a uniformly linearly accelerated worldline 
\begin{equation}
    \gamma(s) = (a^{-1}\sinh as,\ a^{-1}\cosh as,\ 0,\ 0),
\end{equation}
whose velocity evolves under the $1$-parameter group of $tx$-boosts $\dot{\gamma}^\mu(s) = H^\mu_{\phantom{\mu}\nu}(s)\dot{\gamma}^\nu(0)$, where
\begin{equation}\label{eq:H}
    H^\bullet_{\phantom{\bullet}\bullet}(s) = \begin{pmatrix} 
    \cosh as & \sinh as & 0 & 0 \\ 
    \sinh as & \cosh as & 0 & 0 \\ 
    0 & 0 & 1 & 0 \\
    0 & 0 & 0 & 1
    \end{pmatrix} = 
    \exp \begin{pmatrix} 
    0 &  as & 0 & 0 \\ 
    as & 0 & 0 & 0 \\ 
    0 & 0 & 0 & 0 \\
    0 & 0 & 0 & 0
    \end{pmatrix}
\end{equation}
and $0\neq a\in\mathbb{R}$ is fixed.  
Noting that the initial velocity and its first two derivatives are $\dot{\gamma}(0)=\partial_t$, $\ddot{\gamma}(0)=a\partial_x$, $\ddot{\gamma}(0) = a^2\partial_t$, we obtain
an adapted tetrad by choosing the tetrad  
$\partial_t,\partial_x,\partial_y,\partial_z$ at $s=0$, and applying the prescription $e_a^\mu(s) =H^\mu_{\phantom{\mu}\nu}(s)e_a^\nu(0)$ to find
\begin{equation}
    e_0(s) =  \cosh as\partial_t+\sinh as\partial_x, \quad e_1(s) =  \sinh as\partial_t +\cosh as\partial_x, \quad e_2(s) = \partial_y, \quad e_3(s) = \partial_z.
\end{equation} 

Straightforward calculation, following the method of Section~\ref{sec:method}, gives
\begin{equation}
    K(a^2 s^2) = -\frac{3(as)^4}{8\sinh^4(as/2)}
\end{equation}
and hence 
\begin{equation}\label{eq:Tacc}
    T(s,s') = \lim_{\epsilon\to 0+} \frac{3a^4(s-s')^4\cosech^4(a(s-s')/2)}{32\pi^2(s-s'-i\epsilon)^4},
\end{equation}
which may be simplified to
\begin{equation}\label{eq:Tacc2}
    T(s,s')= \lim_{\epsilon\to 0+} \dfrac{3a^4}{32\pi^2}\cosech^4\left(a(s-s'-i\epsilon)/2\right).
\end{equation}
Here, we have used the
general fact that $\lim_{\epsilon\to 0+} g(x)f(x-i\epsilon)=\lim_{\epsilon\to 0+} g(x-i\epsilon)f(x-i\epsilon)$ in the sense of distributions, when  $f$ is analytic in a strip $Z=\{x-iy:x\in\mathbb{R},~0<y<y_0\}\subset\mathbb{C}$ with $\sup_{z\in Z} |f(z) (\Im z)^N|<\infty$ for some $N>0$ and $g$ is analytic on $Z$ and continuous on $Z\cup \mathbb{R}$. 

As the function $K(z)$ evidently decays rapidly as $z\to\infty$ on the real axis, the method of Section~\ref{sec:method} allows us to read off the QEI from the derivatives of $K(z)$ at $z=0$ according to~\eqref{eq:ABfromK}. Using
\begin{equation}
    K(z) = \frac{3z^2}{8\sinh^4 (\sqrt{z}/2)} = -6+z-\frac{11}{120}z^2+O(z^3),
\end{equation}
we find $A=a^2$ and $B=11a^4/30$, in agreement with~\eqref{eq:curvatureinvA} and~\eqref{eq:curvatureinvB} using the invariants from Table~\ref{tab:inv} and reproducing the result~\eqref{eq:QEIlinacc} from~\cite{QEI2006FewsterPfenning}. In that reference, the point-split energy density~\eqref{eq:Tacc2} was found by a direct calculation. Writing $T(s,s')=T(s-s')$, the Fourier transform yields
\begin{equation}
    \hat{T}(u)   =\frac{u^3-a^{2}u}{2\pi(1-e^{-2\pi u /a})}
\end{equation}
and by using the last expression in~\eqref{eq:Qeven}, a calculation gives
\begin{align}
    Q_{\text{even}}(\alpha) 
&= \dfrac{1}{32\pi^3}\left(\alpha^4 + 2a^2\alpha^2 + \dfrac{11}{30}a^4\right),
\end{align}
from which~\eqref{eq:QEIlinacc} follows on inserting the above expression into~\eqref{eq:QEIQeven} and using Parseval's theorem.

\subsection{Hyperbolic subgroups: the catenary} \label{sec:catenary}

Now consider the catenary
\begin{equation} \label{eq:cattraj}
    \gamma(s) = (a^{-1}\cosh\chi\,\sinh as,a^{-1}\cosh\chi\,\cosh as,-s\sinh\chi,0),
\end{equation}
for constant $a\neq 0$, with initial velocity 
\begin{equation}
    \dot{\gamma}^\bullet(0) = (\cosh\chi,0,-\sinh\chi,0),
\end{equation}
and second and third derivatives 
\begin{equation}
    \ddot{\gamma}^\bullet(0) = (0,a\cosh\chi,0,0), \qquad \dddot{\gamma}^\bullet(0) = (a^2\cosh\chi,0,0,0).
\end{equation}
The velocity evolves under the hyperbolic subgroup~\eqref{eq:H}.
Writing $C=\cosh{\chi}$ and $V=\sinh{\chi}$, the tetrad
\begin{align}
    e_0^\bullet(s) &= (C \cosh{a s}, C\sinh{a s }, -V, 0),& \
    e_1^\bullet(s) &= (\sinh{a s}, \cosh{a s}, 0, 0), \nonumber\\
    e_2^\bullet(s) &= (-V\cosh{a s}, -V\sinh{a s}, C, 0),& \
    e_3^\bullet(s) &= (0,0,0,1)
\end{align}
is adapted to $\gamma$ with $e_a^\mu(s) =H^\mu_{\phantom{\mu}\nu}(s)e_a^\nu(0)$. 
A calculation results in the formula
\begin{equation}
    K(z) = -\frac{
    4V^2(\sinhc^2(r)+v^2)\sinh^2(r)+2(4C^2-1)\sinhc^2(r)
    -16V^2\sinhc(2r) + 2v^2(4C^2-3)
    }{C^4(\sinhc^2(r)-v^2)^3}
\end{equation}
where $v=\tanh\chi$, $r=\sqrt{z}/(2\cosh\chi)$ and $\sinhc(x)=\sinh(x)/x$ is the hyperbolic version of the $\sinc$ function.
Note that we need not specify a branch for the square root as it always appears in the argument of an even entire function, and also that $K(z)\to 0$ as $z\to\infty$ in $\mathbb{R}$. 
The series expansion is
\begin{equation}
    K(z) = -6 + \frac{2C^2-1}{C^2} z - \frac{185C^4 - 182C^2+30}{360C^4} z^2 +O(z^3)
\end{equation}
and as $\kappa=a C$ we may read off $A = a^2(2C^2-1) = a^2\cosh{2\chi}$ and $B=(185C^4 - 182C^2+30)a^4/90$. It is straightforward that these values agree with~\eqref{eq:curvatureinvA} and~\eqref{eq:curvatureinvB} using the curvature invariants for this case.  
In particular, the resulting QEI is compatible with a constant negative energy density of \begin{equation}
-T_{\text{reg}}(0)=-\frac{(185\cosh^4\chi-182\cosh^2\chi+30)a^4}{1440\pi^2}
\end{equation}
along the worldline~\eqref{eq:cattraj}. As would be expected, the QEI for linear acceleration 
is obtained in the limit $\chi\to 0$, but for $\chi\neq 0$,
we have $-T_{\text{reg}}(0)< -11a^4/480\pi^2$, and the QEI bound is consistent with a strictly more negative constant energy density than is the case for the linearly accelerated worldline with the same value of $a$.

\subsection{Parabolic subgroups: the semicubical parabola} \label{sec:semicubical}

We now consider the semicubical parabola
\begin{equation}  
    \gamma(s) = \left(s+\dfrac{1}{6}a^2s^3, \dfrac{1}{6}a^2s^3, 0, \dfrac{1}{2}as^2\right),
\end{equation}
for constant $a\neq 0$, whose velocity evolves as $\dot{\gamma}^\mu(s)= P^\mu_{\phantom{\mu}\nu}(s)\dot{\gamma}(0)$ with $\dot{\gamma}(0)=\partial_t$, where  $P^\mu_{\phantom{\mu}\nu}$ was defined in~\eqref{eq:P}. From the initial derivatives $\dot{\gamma}(0)=\partial_t$, $\ddot{\gamma}(0)=a\partial_z$, $\dddot{\gamma}(0)=a^2(\partial_t+\partial_x)$ one sees that 
the initial tetrad 
$e_0(0)=\partial_t$, $e_1(0)=\partial_z$, $e_2(0)=\partial_x$, $e_3(0)=\partial_y$ determines an adapted tetrad 
\begin{align}
     e_0^\bullet(s) &= \left(1+\tfrac{1}{2}(as)^2,  \tfrac{1}{2}(as)^2, 0, as \right), & 
    e_1^\bullet(s) &= \left(as, as, 0, 1 \right), \notag \\
    e_2^\bullet(s) &= \left(-\tfrac{1}{2}(as)^2, 1-\tfrac{1}{2}(as)^2,  0,  -as\right), & 
    e_3^\bullet(s) &= (0, 0, 1,0), 
\end{align}
at general proper time $s$ obeying $e_a^\mu(s)=P^\mu_{\phantom{\mu}\nu}(s) e_a^\nu(0)$.

Straightforward calculation now gives
\begin{equation}
    K(z) = -\frac{6-z/2+5z^2/36}{(1+z/12)^3},
\end{equation}
with 
\begin{equation}
K(z)= -6+2z-\frac{37}{72}z^2+O(z^3)
\end{equation}
as $z\to 0$ and $K(z)=O(z^{-1})$ for $z\to\infty$.
Thus, the point-split energy density is
\begin{equation}\label{eq:Tsplitsemicubical}
    T(s,s') = \lim_{\epsilon\to 0+}\frac{3-a^2(s-s')^2/4+5a^4(s-s')^4/72}{\pi^2(s-s'-i\epsilon)^4(1+a^2(s-s')^2/12)^3}
\end{equation}
and~\eqref{eq:ABfromK} gives $A=2a^2$ and $B=37a^4/18$, in agreement with~\eqref{eq:curvatureinvA} and~\eqref{eq:curvatureinvB}. Thus the QEI along a semicubical parabola is
\begin{equation} \label{eq:QEIfinalsemienergy}
    \int \diff{s} |g(s)|^2\langle\mathbf{:}T_{\mu\nu}\dot{\gamma}^\mu\dot{\gamma}^\nu\mathbf{:}\rangle_\omega(\gamma(s)) \geq -\dfrac{1}{16\pi^2}\int_{-\infty}^\infty \diff{s} \left(\lvert g''(s) \rvert^2 + 4a^2\lvert g'(s) \rvert^2 + \dfrac{37}{18}a^4\lvert g(s) \rvert^2\right),
\end{equation}
for any Hadamard state $\omega$. The long-time scaling limit of the above QEI is then
\begin{equation}
\liminf_{\lambda\xrightarrow{}\infty}\int_{-\infty}^\infty \diff{s} |g_\lambda(s)|^2\langle\mathbf{:}T_{\mu\nu} \dot{\gamma}^\mu\dot{\gamma}^\nu\mathbf{:} \rangle_\omega(\gamma(s))
\ge - \dfrac{37}{288\pi^2}a^4,
\end{equation}
where as usual we choose $g$ with unit $L^2$-norm. The QEI is therefore compatible with 
a constant negative energy density $-37 a^4/(288\pi^2)$ along the semicubical parabola. As one would expect, the QEI
reduces to the inertial case as $a\to 0$.

In fact the QEI~\eqref{eq:QEIfinalsemienergy} can also be obtained by 
a different method. Writing $T(s,s')=T(s-s')$, the Fourier transform may be computed by contour methods as
\begin{equation} \label{eq:fouriernoepsilonsemienergyconcise}
    \hat{T}(u) = \dfrac{1}{2\pi} \left[\left(\dfrac{2u^2}{\sqrt{12}}+\dfrac{7\lvert u\rvert}{8}a+\dfrac{15}{8\sqrt{12}}a^2\right)ae^{-\lvert u\rvert\sqrt{12}/a} + \left(u^3+2ua^2\right)\Theta(u)\right].
\end{equation}
The calculation is considerably simplified if one first replaces powers of
$s-s'$ in~\eqref{eq:Tsplitsemicubical} by powers of $s-s'-i\epsilon$.
To find $Q_{\text{even}}(\alpha)$, we note that
$\hat{T}_{\text{odd}}(u)= (u^3+2ua^2)/(4\pi)$, and also that the integral
of $\hat{T}$ over $(-\infty,0]$ may be evaluated in terms of $\Gamma$-functions. 
After manipulation,  the formula~\eqref{eq:Qeven} gives
\begin{align}
    Q_{\text{even}}(\alpha) &= \dfrac{1}{4\pi^3}\int_{-\infty}^0 \left(\dfrac{2u^2}{\sqrt{12}}+\dfrac{7\lvert u\rvert}{8}a+\dfrac{15}{8\sqrt{12}}a^2\right)ae^{-\lvert u\rvert\sqrt{12}/a} \diff{u} + \dfrac{1}{8\pi^3}\int_0^\alpha \left(u^3+2ua^2\right) \diff{u} \nonumber \\
    &= \dfrac{1}{32\pi^3}\alpha^4 + \dfrac{a^2}{8\pi^3}\alpha^2 + \dfrac{37a^4}{576\pi^3}.
\end{align}
Inserting this expression in~\eqref{eq:QEIQeven} and using Parseval's theorem we reproduce~\eqref{eq:QEIfinalsemienergy}.

\subsection{Elliptic subgroups: uniform rotation} \label{sec:uniformrotation}

Next, consider the uniformly rotating worldline
\begin{equation}
    \gamma(s) = \left(s\cosh\chi, 0, r\cos{\omega s}, r\sin\omega s\right),
\end{equation}
where the radius $r>0$ and proper angular velocity $\omega\neq 0$ together fix the rapidity $\chi = \sinh^{-1}(r\omega)$. In this case,
the velocity evolves under rotations in the $yz$-plane as
$\dot{\gamma}^\mu(s) = R^\mu_{\phantom{\mu}\nu}(
s)\dot{\gamma}^\nu(0)$, where  
\begin{equation}
    R^\bullet_{\phantom{\bullet}\bullet}(s) = \begin{pmatrix} 
    1 & 0 & 0 & 0 \\
    0 & 1 & 0 & 0 \\
    0 & 0 & \cos\omega s & -\sin\omega s \\
    0 & 0 & \sin\omega s & \cos\omega s
    \end{pmatrix} = \exp \begin{pmatrix} 
    0 & 0 & 0 & 0 \\
    0 & 0 & 0 & 0 \\
    0 & 0 & 0 & -\omega s \\
    0 & 0 & \omega s & 0
    \end{pmatrix}.
\end{equation} 
Meanwhile, the initial velocity and its first two derivatives
are
\begin{align*}
    \dot{\gamma}^\bullet(0) &= \left( C, 0, 0,V \right)\\
    \ddot{\gamma}^\bullet(0) &= \left( 0, 0,-V\omega ,0 \right) \\
     \dddot{\gamma}^\bullet(0) &= \left( 0,0,0, -V\omega^2 \right),
\end{align*} 
where we have written $C=\cosh\chi$ and $V=r\omega=\sinh\chi$.
Then $e_0^\bullet(0)=(C,0,0,V)$, $e_1^\bullet(0)=(0,0,-1,0)$,
$e_2^\bullet(0)=(-V,0,0,-C)$, $e_3^\bullet(0)=(0,1,0,0)$,
defines an adapted tetrad at $s=0$, which can be extended
along $\gamma$ by
$e_a^\mu(s)= R^{\mu}_{\phantom{\mu}\nu}(\omega s) e_a^\nu(0)$ to give
\begin{align}
    e_0^\bullet(s) &= (C, 0, -V\sin{\omega s }, V\cos{\omega s}),& 
    e_1^\bullet(s) &= (0, 0, -\cos{\omega s}, -\sin{\omega s}), \nonumber\\
    e_2^\bullet(s) &= (-V, 0, C\sin{\omega s}, -C\cos{\omega s}),& 
    e_3^\bullet(s) &= (0,1,0,0).
\end{align} 
A calculation gives
\begin{equation}
 K(z) =   \frac{ 4C^2 \sin^2(\theta)(1+v^2\sinc^2(\theta)) -2(4C^2-3)v^2\sinc^2(\theta)  
      + 16V^2\sinc(2\theta) +2(4C^2-1) 
 }{C^4 (1-v^2\sinc^2(\theta))^3},
\end{equation}
where $\theta=\sqrt{z}/(2\sinh(\chi))$, with series expansion 
\begin{equation}
    K(z) = -6 + \frac{2\cosh^2\chi-1}{\sinh^2\chi}z -
    \frac{185\cosh^4\chi - 188\cosh^2\chi + 33}{360\sinh^4\chi}z^2 + O(z^3).
\end{equation}
As $\kappa=r\omega^2=\omega\sinh\chi$ we read off
$A=\omega^2\cosh(2\chi)=(2(r\omega)^2+1)\omega^2$ and 
\begin{equation} \label{eq:Tregunirot}
    B = \frac{\omega^4(185 \cosh^4\chi  -188\cosh^2\chi + 33)}{90} = 
     \frac{\omega^4(30 + 182 (r\omega)^2 + 185 (r\omega)^4)}{90},
\end{equation}
which may be substituted into~\eqref{eq:QEIfinalintro} to obtain the QEI in this case. In particular, the QEI is compatible with a constant negative energy density of 
\begin{equation}
    - T_{\text{reg}}(0) =-\frac{\omega^4(30 + 182 (r\omega)^2 + 185 (r\omega)^4)}{1440\pi^2}
\end{equation}
along the worldline. While the point-split energy density may be written down in terms of $K$, we do not know of any closed-form expression for its transform. Thus the method of Sections~\ref{sec:genQEI} and~\ref{sec:method} is the only available way to compute this QEI.

Note that the QEI reduces to the inertial case if $\omega\to 0$ with $r$ fixed -- indeed, even if $r=o(\omega^{-2})$. One might initially be surprised that it does not reduce in the same way
when $r\to 0+$ with $\omega$ fixed. The explanation is that the torsion of the curve does not vanish in this limit, even though the curvature $\kappa$ does. This neatly illustrates the influence of higher
curvature invariants on the QEI bound.

\subsection{Loxodromic subgroups} \label{sec:loxodromic}

Finally, we study the loxodromic worldline
\begin{equation} 
    \gamma^\bullet(s) = (Ca^{-1}\sinh(a s),Ca^{-1}\cosh(a s),V\omega^{-1}\cos(\omega s),V\omega^{-1}\sin(\omega s)),
\end{equation}
where $C=\cosh\chi$, $V=\sinh\chi$ for fixed $\chi\neq 0$, $a\neq 0$ and $\omega\neq 0$. This
worldline undergoes both rotation in the $yz$-plane at constant proper angular velocity $\omega$ and constant distance $|V/\omega|$ from the $x$-axis, while undergoing uniform acceleration in the $x$-direction.
The velocity evolves as $\dot{\gamma}^\mu(s) = L_{a,\omega}{}^\mu_{\phantom{\mu}\nu}(s) \dot{\gamma}^\nu(0)$,
where
\begin{equation}
    L_{a,\omega}{}^\bullet_{\phantom{\bullet}\bullet}(s) = \begin{pmatrix}
    \cosh a s & \sinh a s & 0 & 0 \\
    \sinh a s & \cosh a s & 0 & 0 \\
                        0 & 0   & \cos \omega s & -\sin \omega s \\
                        0 & 0   & \sin \omega s & \cos \omega s
    \end{pmatrix}=
    \exp\begin{pmatrix}
    0 & as & 0 & 0 \\
    as& 0 & 0 & 0 \\
    0 & 0 & 0 & -\omega s \\
    0 & 0 & \omega s & 0
    \end{pmatrix}.
\end{equation}

It can be checked that  
\begin{align}
    e_0^\bullet(s) &= (C\cosh{a s}, C\sinh{a s}, -V\sin{\omega s}, V\cos{\omega s}), \nonumber\\
    e_1^\bullet(s) &= (Ca \kappa^{-1}\sinh{a s}, Ca\kappa^{-1}\cosh{a s}, -V\omega \kappa^{-1}\cos{\omega s}, -V\omega \kappa^{-1}\sin{\omega s}), \nonumber\\
    e_2^\bullet(s) &= (-V\cosh{a s}, -V\sinh{a s}, C\sin{\omega s}, -C\cos{\omega s}), \nonumber\\
    e_3^\bullet(s) &= (V\omega \kappa^{-1}\sinh{a s},V\omega \kappa^{-1}\cosh{a s},Ca \kappa^{-1}\cos{\omega s},Ca \kappa^{-1}\sin{\omega s})
\end{align}
defines an adapted tetrad for $\gamma$, obeying
$e_a^\mu(s)= L_{a,\omega}{}^\mu_{\phantom{\mu}\nu}(s) e_a^\nu(0)$, while the calculation of $K$ by computer algebra produces
\begin{align}
K(z) &= \frac{1}{(C^2\sinhc^2(a r)-V^2\sinc^2(\omega r))^3}\left(
16 C^2 V^2 \sinc(2\omega r)\sinhc(2a r) \right.\notag\\
&\qquad\left.
+4(C^2\sin^2(\omega r)-V^2\sinh^2(a r))(V^2\sinc^2(\omega r)+C^2\sinhc^2(a r)) \right.\notag\\
&\qquad\qquad\left.
-2V^2(C^2+3V^2)\sinc^2(\omega r) - 2C^2(3C^2+V^2)\sinhc^2(a r) 
\right),
\end{align}
where $r=\sqrt{z}/(2\sqrt{C^2a^2+V^2\omega^2})$. For large real $z$, it is easily seen that
\begin{equation}
    K(z)\sim -4V^2 (ar)^2/(C^4\sinhc^2(ar))\to 0
\end{equation}
as $z\to\infty$ in $\mathbb{R}$. 
Meanwhile, the Taylor expansion about $z=0$ reads
\begin{align}
    K(z) &= -6 + \frac{(a^2+\omega^2)(C^2+V^2)}{C^2a^2+V^2\omega^2} z   - \frac{z^2}{360(C^2a^2+V^2\omega^2)^2}\left(
    (33a^4+60a^2\omega^2+30\omega^4)C^4 \right.\notag\\ 
    &\qquad +\left.
    (122a^4+250a^2\omega^2+122\omega^4)(CV)^2
    +(33\omega^4+60a^2\omega^2+30a^4)V^4\right) + O(z^3)\notag
\end{align}
so $A = (C^2+V^2)(a^2 + \omega^2)$, while $B$ is given by
\begin{align}
    90B &= (3a^4+30(a^2+\omega^2)^2)C^4+(3\omega^4+(30(a^2+\omega^2)^2)V^4
    +(122(a^2+\omega^2)^2+6a^2\omega^2)C^2V^2\notag \\
    &=3(C^2a^2+V^2\omega^2)^2 + 62(a^2+\omega^2)^2(CV)^2 + 30 (a^2+\omega^2)^2(C^2+V^2)^2,
    \label{eq:ninetyB}
\end{align}
in which the last term is $30A^2$. These values are easily expressed in terms of 
curvature invariants. Using~\eqref{eq:loxoinv} and $C^2-V^2=1$ one has
\begin{equation}
    \kappa^2(\tau^2+\upsilon^2) = (a^2+\omega^2)^2(CV)^2+(a\omega)^2 = (V^2a^2+C^2\omega^2)(C^2a^2+V^2\omega^2) = \kappa^2(V^2a^2+C^2\omega^2),
\end{equation}
from which the identity 
\begin{equation}
     \kappa^2+\tau^2+\upsilon^2= (a^2+\omega^2)(C^2+V^2) = A
\end{equation}
follows directly, in agreement with~\eqref{eq:curvatureinvA}. Using this in~\eqref{eq:ninetyB} together with~\eqref{eq:loxoinv} we see that $B$ takes the form \eqref{eq:curvatureinvB}. We see that the
QEI is compatible with a constant negative energy density of $-T_{\text{reg}}(0)$ along the worldline~\eqref{eq:loxotraj}, where 
\begin{equation} \label{eq:Tregloxo}
    T_{\text{reg}}(0)=\dfrac{185(a^2 + \omega^2)^2C^4 - (182a^4 + 370a^2\omega^2 + 188\omega^4)C^2 + 33\omega^4 + 60a^2\omega^2 + 30a^4}{1440\pi^2}
\end{equation}
and we have used $V^2=C^2-1$. Note that the QEI does not reduce to the hyperbolic QEI in the limit $\chi\to 0$ with $a$ and $\omega$ fixed. This is because the hypertorsion has a nonzero limit $\upsilon\to \sgn(a)\omega$, even though the torsion vanishes and the curvature tends to $a$. Nonetheless, it is easily seen from~\eqref{eq:ninetyB} that $90B\ge 33a^4$ and hence that $-T_{\text{reg}}(0)<-11a^4/(480\pi^2)$, so that the QEI for loxodromic worldlines can be consistent with a more negative constant value of the energy density than the linearly accelerated worldline with the same value of $a$.

\section{Summary and discussion} \label{sec:discussion}

In this paper we have succeeded in giving an exact closed form expression~\eqref{eq:QEIfinalintro}--\eqref{eq:curvatureinvB} for the QEI for the massless scalar field on any stationary worldline in four-dimensional Minkowski spacetime. This was achieved by a novel method that circumvented the need to take
Fourier transforms of the point-split energy density along the worldline, and which reduced the problem to the computation of certain Taylor coefficients of functions determined by a tetrad adapted to the worldline. In addition, we have given explicit
calculations for the six prototypical classes of stationary trajectory, obtaining agreement with our general result (and also verifying a technical condition needed for the general analysis). The resulting QEI bound depends only on the curvature, torsion and hypertorsion of the worldline. We have also conducted -- in Appendix~\ref{sec:wicksqanalysis} -- a parallel exercise for a quantum inequality on the Wick square. A scaling analysis (see~\eqref{eq:scaling}) shows how these bounds take a universal form on timescales short in relation to the curvature scales, from which they then deviate at longer timescales. In the infinite time limit, they would all allow the field to exhibit a constant negative energy density (or zero in the inertial case). 

Our results complement those of Kontou and Olum~\cite{KontouOlum:2014,KontouOlum:2015}, who computed an absolute QEI~\cite{FewsterSmith:2008} in an approximation of spacetimes where the curvature was weak. There, the worldline was taken to be a geodesic. Our present results indicate the corrections that should enter at leading order when that assumption is dropped. (We reemphasise that our results are exact for massless fields in Minkowski spacetime on stationary trajectories.)

To conclude, we first mention various potential extensions of our work and then return to the question of whether the long-time limits of the QEI are saturated by physical states of the field. Starting with extensions, we expect that our general method would extend fairly directly to stationary worldlines in any even-dimensional Minkowski spacetimes, leading to closed form results in terms of the appropriate curvature invariants. 
In odd dimensions, the vacuum two-point function involves noninteger powers of the geodesic separation, which adds an extra complication. It would be interesting to investigate this case in more detail. 
(For higher-dimensional treatment of the Unruh detector response in higher dimensions, which would be related to the Wick QI in these cases, see~\cite{HodgkinsonLouko:2012}, and for specific calculations relating to the detailed balance definition of Unruh temperature along stationary worldlines in $4$-dimensions, see~\cite{GoodJuarezAubryetal:2020}.) Next, massive fields typically have QEI bounds that are exponentially suppressed relative to the massless ones. Here, we do not expect that our method would easily produce closed-form results, but again, it would be worth investigating, as would the situation for higher spin fields. 
  
Finally, we consider the extent to which the long term average bounds can be attained. In the case of inertial worldlines this is obvious: the long-term average value of zero is attained in the Minkowski vacuum state. 
For uniformly accelerated curves it was noted in~\cite{QEI2006FewsterPfenning} that the   bound~\eqref{eq:AWECacc} is attained by the Rindler vacuum for the right wedge $x>|t|$ in Minkowski spacetime. It is useful to put this in a broader context. Adopting coordinates 
$t=\xi\sinh \chi$, $x=\xi\cosh \chi$, the Rindler wedge $x>|t|$ of Minkowski spacetime has metric $\xi^2\diff{\chi}^2 -\diff{\xi}^2-\diff{y}^2-\diff{z}^2$,
and any curve $\chi\mapsto (a\chi,1/a,y_0,z_0)$ with $a>0$ is a curve of proper acceleration $a$ in proper time parameterisation. Moreover, the energy density measured by an observer moving on a curve of constant $\xi$, in the thermal state of temperature $\beta^{-1}$ with respect to the coordinate $\chi$, is  
\begin{equation}\label{eq:Rindler_beta}
    \langle {:}T_{\mu\nu}u^\mu u^\nu{:}\rangle_\beta=
    \frac{(4\pi^2-\beta^2)(33\beta^2+12\pi^2)}{1440\pi^2\beta^4 \xi^4},
\end{equation}
reducing to  
\begin{equation}
  \langle {:}T_{\mu\nu}u^\mu u^\nu{:}\rangle_\infty = -\frac{11}{480\pi^2\xi^4}
\end{equation} 
for the Rindler ground state. At $\beta=2\pi$, the thermal state on Rindler spacetime is precisely the 
restriction of the Minkowski vacuum to the right wedge, which is why the energy density vanishes.
Because most references (e.g.,~\cite{CandelasDeutsch:1977,Dowker:1978,BrownOttewillPage:1986}) only discuss the conformally coupled stress-energy tensor (the `new improved' stress tensor) and~\cite{QEI2006FewsterPfenning} only considered the ground state without giving details, the relevant calculations are briefly reviewed in Appendix~\ref{sec:Rindler}. On restriction to the curve $\xi=1/a$ we see that
all these states have constant energy density
consistent with~\eqref{eq:AWECacc} (see Fig.~\ref{fig:rhobeta}) and that this bound is attained by the Rindler ground state. 
\begin{figure}[t]
\begin{center}
\begin{tikzpicture}[scale=0.4]
\draw [color=red,domain=0.8:8,samples=40,variable=\x] plot ({3*\x},{(1/\x^2-1)*(11+1/\x^2)});
\draw [->] (0,0)--++(3*8,0) node[below]{\raisebox{-10pt}{$\beta$}};
\draw [->] (0,-12)--(0,8) node[left]{$\rho$};
\draw[dotted] (0,-11)-++(3*8,0);
\draw (0,-11) node[left]{$-11$};
\draw (0,0) node[left]{$0$};
\foreach \x in {2,4,6,8,10,12,14}
    {\draw ({3*\x/2},0) node[below]{\raisebox{-10pt}{$\x\pi$}} --++(0,-0.4);}
\end{tikzpicture}
\end{center}
\caption{Plot of $\rho=(480\pi^2\xi^4)\langle {:}T_{\mu\nu}u^\mu u^\nu{:}\rangle_\beta$ on a curve of constant $\xi$, against $\beta$. The dotted line corresponds to the QEI bound~\eqref{eq:AWECacc}, which is attained as $\beta\to\infty$, corresponding to the Rindler ground state.}\label{fig:rhobeta}
\end{figure}
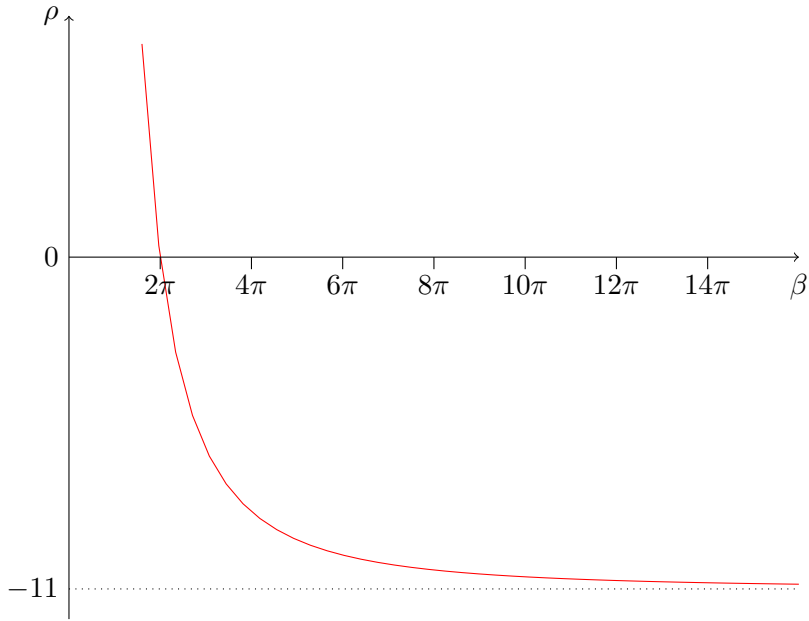

One should note that the Rindler ground state (and indeed all the $\beta$-KMS states other than the special case $\beta=2\pi$) is not defined on all of Minkowski, but just on the wedge $x>|t|$. The obvious divergence of the stress-energy tensor as $\xi\to 0+$ shows that the state cannot be extended as a Hadamard state beyond the wedge. The reason they satisfy the
Minkowski QEI is because this QEI is local and covariant -- see~\cite{QEI2006FewsterPfenning} for a discussion and many similar calculations, and~\cite{Fewster2007} for a more abstract viewpoint inspired by~\cite{BrFrVe03}. 
Nonetheless, it remains open as to whether equality in~\eqref{eq:AWECacc} can be attained by a Hadamard state defined on all of Minkowski; our conjecture is that one can find global Hadamard states that approximate the Rindler ground state sufficiently well that the bound~\eqref{eq:AWECacc} is satisfied in a limiting sense. These issues will be addressed elsewhere.  

Turning to the remaining stationary worldlines, the QEI is again consistent with a constant strictly negative energy density and we can again ask whether the bound is attained in any sense. Letaw and Pfautsch~\cite{LetawPfautsch:1981} considered the problem of quantising the field in coordinates associated with the various stationary worldlines and seeking an appropriate ground state. For the inertial, uniformly rotating, and semicubical parabolic worldlines, they concluded that the resulting state was precisely the Minkowski vacuum state. This means that we have no obvious candidate state associated with the uniformly rotating and semicubical parabolic worldlines with negative energy density. On the other hand, the catenary~\eqref{eq:cattraj} and loxodromic worldlines~\eqref{eq:loxotraj} both result in a Rindler vacuum state on the $x>|t|$ wedge, which is the causal hull of the worldline in question. One may compute the energy density along these curves in the Rindler vacuum, using the renormalised stress energy tensor given in Appendix~\ref{sec:Rindler}, yielding constant energy densities
$-(14 \cosh^2\chi+19)a^4/(1440\pi^2 \cosh^4\chi)$ in each case. 
This value is strictly greater than $-11a^4/(480\pi^2)$ for $\chi\neq 0$, which is greater than the most negative constant energy density consistent with the QEIs in these cases (see the remarks at the end of sections~\ref{sec:catenary} and~\ref{sec:loxodromic}). Thus they are are consistent with the QEIs but do not saturate them.

It therefore remains an open and intriguing question, whether there are (sequences of) Hadamard states that attain these QEI bounds (in a limiting sense). Resolving this question, and its analogues in 2+1 dimensions, may have relevance to proposed experiments to detect the Unruh effect using a laser beam whose intersection with a 
Bose-Einstein condensate follows a uniformly rotating worldline~\cite{Gooding_etal_PRL:2020}.

\bigskip
{\noindent\emph{Acknowledgements}
CJF thanks Alexander Strohmaier and Valter Moretti for useful conversations concerning the H\"ormander pseudo-topologies, and Aron Wall for posing an interesting direction for further study. The work of JT was in part funded by an EPSRC studentship at the University of Sheffield and a summer studentship from the University of York. We thank Elizabeth Winstanley for a reading of the manuscript and some helpful suggestions.} 
\appendix
\section{Details on the method}\label{sec:details}

We give further details on the method described in Section~\ref{sec:method} and prove the Lemma stated there.
Some aspects are treated using techniques of microlocal analysis -- we will be rather brief on those details, referring the reader to appropriate literature, while indicating the structure of the argument.

To start, we observe that, for $\epsilon>0$, $F(\sigma_\epsilon(x,0))$
can be written
\begin{equation}
F(\sigma_\epsilon(x,0)) = \int\frac{\mathrm{d}^3\kb}{(2\pi)^3}
\frac{e^{-\|\kb\|\epsilon-ik\cdot x}}{2\|\kb\|},
\end{equation}
where $k_\bullet=(\|\kb\|,\kb)$, $x^\bullet=(t,\xb)$. Thus 
for any $\varphi\in C_0^\infty(\mathbb{R}^4)$, the distribution
$u_\epsilon(x) = \varphi(x) F(\sigma_\epsilon(x,0))$ has Fourier transform
\begin{equation}
\hat{u}_\epsilon(k') = \int\frac{\mathrm{d}^3\kb}{(2\pi)^3}
\frac{e^{-\epsilon\|\kb\|}}{2\|\kb\|} \hat{\varphi}(k'-k).
\end{equation}
As $\hat{\varphi}$ decays faster than inverse polynomials and
$k\in \mathcal{N}^+$, where $\mathcal{N}^{+/-}$ is the bundle of future/past-pointing null covectors, it may be shown that
$F(\sigma_\epsilon(x,0))$ converges in $\mathscr{D}'_{\mathcal{N}^+}(\mathbb{R}^4)$ with respect to the H\"ormander pseudo-topology~\cite{Hoermander1}. It follows from this that
the vacuum $2$-point function $G_0(x,x')$ is the limit of
$F(\sigma_\epsilon(x,x'))=F(\sigma_\epsilon(x-x',0))$ 
in $\mathscr{D}'_{\mathcal{N}^+\times\mathcal{N}^-}(\mathbb{R}^4\times\mathbb{R}^4)$ and has wavefront set 
$\WF(G_0)\subset \mathcal{N}^+\times \mathcal{N}^-$, as is also 
known on general grounds because the state is Hadamard~\cite{Radzikowski:1996}.

These facts have various consequences. First, the pull-back of (any derivative operator acting on) $G_0$ by $\varphi:(s,s')\mapsto(\gamma(s),\gamma(s'))$ is well-defined because the set of normals to $\varphi$ does not intersect $\WF(G_0)$, essentially because timelike and null vectors cannot be orthogonal -- see~\cite{GenQEI2000CFewster} for details. Consequently the pull-back is well-defined by standard results explained in Chapter 8 of~\cite{Hoermander1} and has
wavefront set contained in $\varphi^*\WF(G_0)\subset \varphi^* (\mathcal{N}^+\times\mathcal{N}^-) =\Gamma\times (-\Gamma)$, where $\Gamma= \mathbb{R}\times (0,\infty) \subset T^*\mathbb{R}$.
Moreover, $\varphi^*G_0$ is the limit in $\mathscr{D}'_{\Gamma\times (-\Gamma)}(\mathbb{R}\times\mathbb{R})$ of $\varphi^* F\circ\sigma_\epsilon$ as $\epsilon\to 0+$, which justifies taking the pull-back under the $\epsilon\to 0+$ limits in~\eqref{eq:diff2G0}. 
Similar arguments apply to the convergence of $F'(\sigma_\epsilon(x,x'))$ and $F''(\sigma_\epsilon(x,x'))$ 
as $\epsilon\to 0+$.

Next, recall that the stationary worldline $\gamma$ has velocity $u=\dot{\gamma}$ evolving according to $u(s) =\exp(sM)u(0)$, for $M\in\mathfrak{so}(1,3)$ with dimensions of inverse time, and that the right-handed tetrad $e_a(s)$ obeys $e_a(s)=\exp(sM) e_a(0)$, with  $u(s)=e_0(s)$,  $\dot{u}(s)\in \Span \{e_1(s)\}$, 
and $\ddot{u}(s) \in  \Span \{e_0(s),e_1(s), \linebreak[0] e_2(s)\}$. 
The Cartesian coordinates of $\gamma(s)$, and components
of $e_a(s)$ are evidently real analytic in $s$. We extend $e_a$ to a smooth tetrad in a neighbourhood of $\gamma$ in an arbitrary fashion. Recall that the functions $C_a$ and $D_a$ are defined, in index-free notation, by
\begin{equation}
    C_a(s,s') = \eta( e_a(s),e_a(s')), \qquad
    D_a(s,s') = \eta(\gamma(s)-\gamma(s'), e_a(s)).
\end{equation}

We now prove the lemma needed in Section~\ref{sec:method}, which we restate for convenience.  
\begin{lemma}
(a) With the choice of tetrad just described, $C_a(s,s')$ and $D_a(s,s')$ are translationally invariant,
depending only on $s-s'$. There are entire analytic functions $G_a$ and $H_a$ such that
\begin{equation}
C_a(s,s')=G_a(\kappa^2(s-s')^2), \qquad D_a(s,s')D_a(s',s) = -(s-s')^2 H_a(\kappa^2(s-s')^2),
\end{equation}
where in the limit $z\to 0$,
\begin{equation}
    \sum_{a=0}^3 G_a(z)  = -2 + \frac{\tau^2+\upsilon^2}{\kappa^2}z + 
\frac{(\kappa\tau)^2- (\tau^2+\upsilon^2)^2}{\kappa^4} z^2 + \mathcal{O}(z^3),
\end{equation}
and
\begin{equation}
\sum_{a=0}^3 H_a(z) = 1 + \frac{z}{12} + \frac{\kappa^2+19\tau^2}{360\kappa^2}z^2 +\mathcal{O}(z^3).
\end{equation}

(b) The signed square geodesic separation of points along $\gamma$ obeys
\begin{equation}\label{eq:Upsilon}
    \sigma_0(\gamma(s),\gamma(s')) = -(s-s')^2 \Upsilon(\kappa^2(s-s')^2),
\end{equation}
where $\Upsilon$ is entire analytic with $\Upsilon(z) = 1 + \tfrac{1}{12}z + \tfrac{1}{360}(1-\tau^2/\kappa^2)z^2+\mathcal{O}(z^3)$ 
as $z\to 0$. Furthermore, for  $z\in [0,\infty)$, $\Upsilon(z)$ is real with $\Upsilon(z)\ge 1$. 
\end{lemma}
\begin{proof}
(a) For inertial worldlines, $e_a(s)$ is constant and the result holds trivially with $G_0(z)\equiv 1$, $G_i(z)\equiv 1$, $H_0(z)\equiv -1$, $H_i(z)\equiv 0$.
From now on we may assume that $\kappa$ is nonzero.

It follows from~\eqref{eq:eadef} that 
\begin{equation}
    C_a(s,s')=\eta(\exp(s'M)e_a(0),\exp(sM)e_a(0)) = \eta(e_a(0),\exp((s-s')M)e_a(0)),
\end{equation}
so $C_a$ depends only on $s-s'$. As every component
of the matrix $\exp(sM)$ is analytic, and because
$C_a(s,s')=C_a(s',s)$, we deduce that $C_a(s,s')=G_a(\kappa^2(s-s')^2)$ for dimensionless entire analytic functions $G_a$. 

Next, observe that 
\begin{equation}\label{eq:Da_deriv}
\frac{\partial}{\partial s'}D_a(s,s') = - \eta(e_0(s'),e_a(s)) = -\eta(e_0(0),\exp((s-s')M)e_a(0)).
\end{equation}
Integrating with respect to $s'$ and using $D_a(s,s) = 0$ we may deduce that $\kappa D_a(s,s')$ is a dimensionless entire analytic function of $(s-s')\kappa$.
Again using $D_a(s,s)=0$ and because \eqref{eq:Da_deriv} gives $\partial D_0/\partial s'|_{s'=s}=-1$ and $\partial D_i/\partial s'|_{s'=s}=0$ for $i=1,2,3$, we have 
\begin{equation}
    D_0(s,s') = (s-s')\left(1 + \mathcal{O}((\kappa(s-s'))^2)\right), \qquad
    D_i(s,s') = \kappa^{-1}\mathcal{O}((\kappa(s-s'))^2),
\end{equation}
where we have also used the fact that $D_0(s,s')=-D_0(s',s)$ as a consequence of~\eqref{eq:Da_deriv}. Because $D_a(s,s')D_a(s',s)$ is invariant under interchange of $s$ and $s'$, we now have
\begin{equation}
    D_a(s,s')D_a(s',s) =-(s-s')^2 H_a(\kappa^2 (s-s')^2)
\end{equation}
for dimensionless entire analytic functions $H_a$.
The Taylor series of $G_a$, $H_a$ and their sums, are computed up to second order in Appendix~\ref{sec:Taylor}. 

(b) Next, we study the geodesic separation between $\gamma(s)$ and $\gamma(s')$. We note that 
\begin{equation}
\frac{\partial}{\partial s}   \sigma_0(\gamma(s),\gamma(s')) = -2D_0(s,s')
\end{equation}
depends only on $s-s'$, so $\sigma_0(\gamma(s),\gamma(s')) = \Sigma(s-s')+ f(s')$
and on considering $s=s'$ we find that $f$ is constant and may be absorbed into $\Sigma$, which is also seen to be even. The first terms in its Taylor expansion are easily found: $\Sigma(0)=0$, while 
\begin{equation}
    \Sigma''(s-s') = -2\eta(u(s),u(s')), \qquad
    \Sigma^{(4)}(s-s') = 2\eta(\dot{u}(s),\dot{u}(s')), \qquad
    \Sigma^{(6)}(s-s') = -2\eta(\ddot{u}(s),\ddot{u}(s'))
\end{equation}
giving
\begin{equation}
    \Sigma''(0) = -2, \qquad \Sigma^{(4)}(0)= -2\kappa^2, \qquad
    \Sigma^{(6)}(0) = -2\kappa^2(\kappa^2-\tau^2)
\end{equation}
using~\eqref{eq:e0_derivs}.
Accordingly, we have established~\eqref{eq:Upsilon}, the analyticity of $\Upsilon$, and also the expansion 
\begin{equation}
    \Upsilon(z) = 1 + \frac{z}{12} + \frac{\kappa^2-\tau^2}{360\kappa^2} z^2 + 
    \mathcal{O}(z^3)
\end{equation}
as $z\to 0$.
Finally, as $\gamma(0)$ and $\gamma(s)$ are connected by a smooth timelike curve, the timelike geodesic that connects them maximises proper time. 
Thus $-\sigma_0(\gamma(s),\gamma(0))\ge s^2$ for all $s\in \mathbb{R}$ and consequently, $\Upsilon(z)\ge 1$ for $z\in [0,\infty)$, which concludes the proof.
\end{proof}

Finally, we explain how the identity~\eqref{eq:keyfact} may be proved. First note that 
\begin{align*}
    \sigma_\epsilon(\gamma(s),\gamma(s')) &= \sigma_0(\gamma(s),\gamma(s')) + 2i\epsilon (\gamma^0(s)-\gamma^0(s')) +\epsilon^2 \\ &= 
    -(s-s')^2 \Upsilon(\kappa^2(s-s')^2) + 2i\epsilon (\gamma^0(s)-\gamma^0(s')) +\epsilon^2 \\
    &=  -(s-s'-i\epsilon)^2 \Upsilon(\kappa^2(s-s')^2)+ \epsilon \Psi(s,s')+\epsilon^2\Xi(s,s')
\end{align*}
for smooth (indeed analytic) functions $\Psi$ and $\Xi$. 
Let $S$ be the difference between the distribution on the left-hand side of~\eqref{eq:keyfact} and the distribution on the right-hand side. Then, using the fact that $\Upsilon$ is nonvanishing on the real axis, $S$ takes the form
\begin{equation}\label{eq:Sform}
    S(s,s') = \lim_{\epsilon\to 0+}\sum_{r=1}^{2k} \frac{ 
     \epsilon^r S_r(s,s')
    }{\sigma_\epsilon(\gamma(s),\gamma(s'))^k (s-s'-i\epsilon)^{2k}}
\end{equation}
for smooth functions $S_r\in C^\infty(\mathbb{R}^2)$ ($1\le r\le 2k$). All that is needed now is to show that the distributional limit
\begin{equation}
 \lim_{\epsilon\to 0+} \frac{1}{\sigma_\epsilon(\gamma(s),\gamma(s'))^k (s-s'-i\epsilon)^{2k}} 
\end{equation}
exists, whereupon $S$ must vanish due to the strictly positive powers of $\epsilon$ in~\eqref{eq:Sform}. The required result now
follows from the sequential continuity of the distributional product
with respect to the H\"ormander pseudo-topology (Theorem 2.5.10 in~\cite{Hormander_FIOi}), and the fact that
both $1/\sigma_\epsilon(\gamma(s),\gamma(s'))$ and 
\begin{equation}
\frac{1}{s-s'-i\epsilon} = i\int_0^\infty \diff{k}\, e^{-ik(s-s'-i\epsilon)}
\end{equation}
have limits as $\epsilon\to 0+$ in $\mathscr{D}_{\Gamma\times (-\Gamma)}'(\mathbb{R}^2)$,
where, as before, $\Gamma= \mathbb{R}\times (0,\infty) \subset \dot{T}^*\mathbb{R}$.

\section{Taylor series calculation} \label{sec:Taylor}
We compute the Taylor series of both $G_a$ and $H_a$ 
up to second order, using equations~\eqref{eq:CandDdefn}, \eqref{eq:CtoGandDtoH} and \eqref{eq:CandDaltdefn}. Recalling that $C_a(s,s') = G_a(\kappa^2(s-s')^2)$, one can expand the right hand side into a Taylor series in $s-s'$ about the point $s-s'=0$ and then differentiate to yield
\begin{align}
    -\frac{1}{2\kappa^2}\frac{\partial^2C_a}{\partial s\partial s'} &= 
    G'_a(0) + 3\kappa^ 2(s-s')^2 G_a''(0) + \mathcal{O}((s-s')^4) \\
    \frac{1}{12\kappa^4}\frac{\partial^4C_a}{\partial^2 s\partial^2 s'} &= 
      G_a''(0) + \mathcal{O}((s-s')^2)
\end{align}
as $s-s'\to 0$.
Differentiating equation~\eqref{eq:CandDdefn} and setting $s=s'=0$, one easily finds
\begin{equation}
    G'_a(0) = -\frac{\eta(\dot{e}_a(0),\dot{e}_a(0))}{2\kappa^2},\qquad
    G''_a(0) = \frac{\eta(\ddot{e}_a(0),\ddot{e}_a(0))}{12\kappa^4}
\end{equation} 
by equating powers of $s-s'$.
The derivatives of the $e_a$ can be read off from the generalized 
Frenet-Serret equations~\eqref{eq:FS} and its derivatives~\eqref{eq:tetradderivatives}, allowing us to express $G_a'(0)$ and $G_a''(0)$ in terms of curvature invariants. 

An easy computation shows that
\begin{equation}
    G_a'(0) = \dfrac{1}{2}\eta_{0a}+\dfrac{\kappa^2-\tau^2}{2\kappa^2}\eta_{1a}-\dfrac{\tau^2+\upsilon^2}{2\kappa^2}\eta_{2a}-\dfrac{\upsilon^2}{2\kappa^2}\eta_{3a}
\end{equation}
and 
\begin{equation}
    G_a''(0) = \dfrac{\kappa^2-\tau^2}{12\kappa^2}\eta_{0a}+\dfrac{\tau^2\upsilon^2+(\kappa^2-\tau^2)^2}{12\kappa^4}\eta_{1a}-\dfrac{\kappa^2\tau^2-(\tau^2+\upsilon^2)^2}{12\kappa^4}\eta_{2a}+\upsilon^2\dfrac{\tau^2+\upsilon^2}{12\kappa^4}\eta_{3a},
\end{equation}
where $\eta(e_a(0),e_b(0))=\eta_{ab}$ by orthogonality of the tetrad field.
Reconstructing $G_a$ using a Taylor series therefore yields
\begin{align}
    G_a(z) &= \eta_{aa} + \dfrac{1}{2\kappa^2}z\left(\eta_{0a}\kappa^2-\eta_{1a}(\tau^2-\kappa^2)-\eta_{2a}(\upsilon^2+\tau^2)-\eta_{3a}\upsilon^2\right)
    \notag\\
    &\quad +\dfrac{z^2}{24\kappa^4}\left(\eta_{0a}\kappa^2(\kappa^2-\tau^2)+\eta_{1a}(\tau^2\upsilon^2+(\kappa^2-\tau^2)^2)-\eta_{2a}(\kappa^2\tau^2-(\tau^2+\upsilon^2)^2)+\eta_{3a}\upsilon^2(\tau^2+\upsilon^2)\right)\notag\\
    &\quad +\mathcal{O}(z^3).
\end{align}
Summing, we obtain
\begin{equation}
    \sum_{a=0}^3 G_a(z)  = -2 + \frac{\tau^2+\upsilon^2}{\kappa^2}z + 
\frac{(\kappa\tau)^2- (\tau^2+\upsilon^2)^2}{\kappa^4} z^2 + \mathcal{O}(z^3)
\end{equation}
as $z\to 0$.

Applying exactly the same methodology to $H_a$, one writes $E_a(s,s') = D_a(s,s')D_a(s',s)$ so that
\begin{align}
E_a(s,s') &= -(s-s')^2 H_a(\kappa^2(s-s')^2) \notag\\ 
&= 
-(s-s')^2 H_a(0) - \kappa^2(s-s')^4 H_a'(0) - \tfrac{1}{2}\kappa^4 (s-s')^6 H_a''(0) + O((s-s')^8).
\end{align}
Differentiation yields
\begin{align}
    \dfrac{\partial^2 E_a}{\partial s\partial s'} &= 2H_a(0)+12\kappa^2(s-s')^2H_a'(0)+ 15\kappa^4 (s-s')^4 H_a''(0)+ \mathcal{O}((s-s')^6) \\
    \dfrac{\partial^4 E_a}{\partial^2 s\partial^2 s'} &= -24\kappa^2 H_a'(0) -
    180\kappa^4(s-s')^2 H_a''(0)+ \mathcal{O}((s-s')^4) \\
    \dfrac{\partial^6 E_a}{\partial^3 s\partial^3 s'} &= 360\kappa^4 H_a''(0)+ \mathcal{O}((s-s')^4),
\end{align}
from which $H_a(0)$, $H_a'(0)$ and $H_a''(0)$ can be obtained differentiating equation~\eqref{eq:CandDaltdefn} using Leibniz' rule and subsequently setting $s=s'=0$.
It is easily verifiable that this yields
\begin{align}
    H_a(0) &= \left[\eta(\dot{\gamma}(0),e_a(0))\right]^2 = \left[\eta(e_0(0),e_a(0))\right]^2, \\
    H_a'(0) &= -\dfrac{1}{4\kappa^2}[\eta(\ddot{\gamma}(0),e_a(0))]^2+\dfrac{1}{3\kappa^2}\eta(\dot{\gamma}(0),e_a(0))\eta(\dddot{\gamma}(0),e_a(0)), \\
    H_a''(0) &= \dfrac{1}{18\kappa^4}[\eta(\dddot{\gamma}(0),e_a(0))]^2-\dfrac{1}{12\kappa^4}\eta(\ddot{\gamma}(0),e_a(0))\eta(\gamma^{(4)}(0),e_a(0))\notag \\&
    \qquad+\dfrac{1}{30\kappa^4}\eta(\dot{\gamma}(0),e_a(0))\eta(\gamma^{(5)}(0),e_a(0)),
\end{align}
and after some straightforward computation,
\begin{align}
    H_{a}(0) &= \eta_{0a} \\
    H_a'(0) &= \dfrac{1}{3}\eta_{0a}+\dfrac{1}{4}\eta_{1a} \\
    H_a''(0) &= (\eta_{0a})^2\left(\dfrac{1}{18}+\dfrac{\kappa^2-\tau^2}{30\kappa^2}\right)- \dfrac{\kappa^2-\tau^2}{12\kappa^2}(\eta_{1a})^2+ \dfrac{\tau^2}{18\kappa^2}(\eta_{2a})^2 \notag \\
    &= \eta_{0a}\left(\dfrac{1}{18}+\dfrac{\kappa^2-\tau^2}{30\kappa^2}\right)+ \dfrac{\kappa^2-\tau^2}{12\kappa^2}\eta_{1a}- \dfrac{\tau^2}{18\kappa^2}\eta_{2a}
\end{align}
using the fact that $(\eta_{0a})^2 = \eta_{0a}$ and $(\eta_{ia})^2 = -\eta_{ia}$ for $i=1,2,3$, as can be explicitly seen in the calculation of $H_a''(0)$.
Reconstructing $H_a$ using a Taylor series, one obtains
\begin{multline}
    H_a(z) = \eta_{0a}+\dfrac{1}{12}z\left(4\eta_{0a}+3\eta_{1a}\right)\\+\dfrac{1}{360\kappa^2}z^2\left(\eta_{0a}(10\kappa^2+6(\kappa^2-\tau^2))+15\eta_{1a}(\kappa^2-\tau^2)-10\eta_{2a}\tau^2\right)+\mathcal{O}(z^3),
\end{multline}
and summing,
\begin{equation}
\sum_{a=0}^3 H_a(z) = 1 + \frac{z}{12} + \frac{\kappa^2+19\tau^2}{360\kappa^2}z^2 +\mathcal{O}(z^3).
\end{equation}

\section{Wick square} \label{sec:wicksqanalysis}

In this Appendix we show how a quantum inequality for the Wick square can be obtained along stationary trajectories. This is a simpler calculation than the one used for the energy density and we shall be relatively brief.

Recall that the general QEI involves a (sum of) pull-backs of a suitable differential operator acting on the two-point function,  
\begin{equation}
    T(s,s') = \langle \mathcal{Q}\phi(\gamma(s))\mathcal{Q}\phi(\gamma(s'))\rangle_{\omega_0} = ((\mathcal{Q}\otimes\mathcal{Q})G_0)(\gamma(s),\gamma(s')).
\end{equation}
For a quantum inequality on the Wick square, the operator $\mathcal{Q}$ can be simply identified as the identity, so $T(s,s')$ can be written in this case as
\begin{equation}
    T(s,s') = G_0(\gamma(s),\gamma(s')).
\end{equation}
Using the results of Section \ref{sec:method} and in particular, equation \eqref{eq:keyfact}, the two-point function can be neatly expressed as
\begin{equation}
    T(s,s') = \lim_{\epsilon\to 0+}\dfrac{1}{4\pi^2\sigma_{\epsilon}(\gamma(s),\gamma(s'))} = -\lim_{\epsilon\to 0+}\dfrac{1}{4\pi^2(s-s'-i\epsilon)^2}\left[\Upsilon\left(\kappa^2(s-s')^2\right)\right]^{-1}.
\end{equation}
As $\Upsilon(\kappa^2 s^2) \ge 1$ for $s\in\mathbb{R}$ by the Lemma, the entire function $\Upsilon(z)$ is nonvanishing on the real axis, and $\Upsilon(z)^{-1}$ is therefore analytic in a neighbourhood of the real axis. Using~\eqref{eq:upsilongeneral} we may write $\Upsilon(z)^{-1} = 1 + zJ(z)$, where $J$ is also analytic in a neighbourhood of the real axis, with $J(0)=-1/12$.
Because $0<1+zJ(z)\le 1$ for $z\ge 0$, we may deduce that $0\le -J(z)<1/z$ for $z>0$.

We can now split 
the pulled back two-point function into its singular and regular parts as $T(s,s') = T_{\text{sing}}(s-s')+T_{\text{reg}}(s-s')$, where 
\begin{equation}
    T_{\text{sing}}(s) = -\dfrac{1}{4\pi^2}\lim_{\epsilon\to 0+}\dfrac{1}{(s-i\epsilon)^2},
\end{equation}
and
\begin{equation}  
T_{\text{reg}}(s) = -\dfrac{J(\kappa^2 s^2)}{4\pi^2}\lim_{\epsilon\to 0+}\dfrac{\kappa^2 s^2}{(s-i\epsilon)^2}=  -\dfrac{\kappa^2 J(\kappa^2 s^2)}{4\pi^2} ,
\end{equation}
with $T_{\text{reg}}(0) = \kappa^2/(48\pi^2)$. Here we used the identity $\lim_{\epsilon\to 0+} x^2/(x-i\epsilon)^2=\lim_{\epsilon\to 0+} (x-i\epsilon)^2/(x-i\epsilon)^2=1$ of distributional limits, because $g(z) = z^2$ is entire, while $f(z) = z^{-2}$ is analytic in the open lower half-plane $Z\subset \mathbb{C}$ and obeys $\sup_{z\in Z}\lvert f(z) (\Im z)^2\rvert =1$ (see the argument below equation \eqref{eq:Tacc2}).

Observing that the two-point function given above is translationally invariant, we can use the bound given by \eqref{eq:genQEI2} and \eqref{eq:Qboundinitial} and thus write
\begin{equation} \label{eq:QEIQevenWick}
    \int \diff{s} |g(s)|^2\langle\mathbf{:}(\mathcal{Q}\phi)^2\mathbf{:}\rangle_\omega(\gamma(s)) \geq - \int_{-\infty}^\infty \diff{\alpha} \lvert \hat{g}(\alpha) \rvert^2 Q_{\text{even}}(\alpha)
\end{equation}
where 
\begin{equation} \label{eq:QevenWick}
    Q_{\text{even}}(\alpha) = \dfrac{1}{2\pi^2}\left[\int_{-\infty}^0 \hat{T}(u)\diff{u}+\int_{0}^\alpha \hat{T}_{\text{odd}}(u)\diff{u}\right].
\end{equation}
The Fourier transform of $T_{\text{sing}}$ is easily shown to be $\hat{T}_{\text{sing}}(u)=\tfrac{u}{2\pi}\Theta(u)$. Again, $T_{\text{reg}}$ is smooth, real and even on $\mathbb{R}$, decaying like $\mathcal{O}(s^{-2})$ as $\lvert s \rvert \to \infty$ because of the decay of $J$. Evidently $T_{\text{reg}}$ does not contribute to $\hat{T}_{\text{odd}}$ as $T_{\text{reg}}$ is absolutely integrable and has a well defined, continuous, real and even Fourier transform. In this case, $T_{\text{sing}}$ is actually universal; the information relating to the specific worldline is encoded in $T_{\text{reg}}$, as can also be seen below in Eq.~\eqref{eq:QIfinalgeneral}. Clearly, $\hat{T}_{\text{sing}}$ does not contribute to the first term in~\eqref{eq:QevenWick} and, recalling that $T_{\text{reg}}$ is even, the odd part of $\hat{T}$ is
\begin{equation}
    \hat{T}_{\text{odd}}(u) = \dfrac{u}{4\pi},
\end{equation}
and so $Q_{\text{even}}$ is given in the form
\begin{align}
    Q_{\text{even}}(\alpha) &= \dfrac{1}{2\pi^2} \left[ \int_{-\infty}^0 \diff{u} \,\hat{T}_{\text{reg}}(u)  + \dfrac{1}{4\pi}\int_0^\alpha  \diff{u} \ u \right] 
    = \dfrac{1}{16\pi^3}\alpha^2 + \dfrac{T_{\text{reg}}(0)}{2\pi}.
\end{align}
In direct analogy to the analysis of the energy density, the evenness of $T_{\text{reg}}$ and the Fourier inversion formula have been used. Inserting this into~\eqref{eq:QEIQevenWick} gives the QI bound
\begin{equation} \label{eq:QIfinalgeneral}
    \int \diff{s} |g(s)|^2\langle\mathbf{:}\phi^2\mathbf{:} \rangle_\omega(\gamma(s)) \geq -\dfrac{1}{8\pi^2}\int_{-\infty}^\infty \diff{s} \left(\lvert g'(s) \rvert^2 + C \lvert g(s) \rvert^2\right).
\end{equation}
where $C = 8\pi^2T_{\text{reg}}(0)=\kappa^2/6$.

Considering the scaling behaviour, using the same test function $g_{\lambda}(s)=\lambda^{-1/2}g(\lambda/s)$ as in the case for the QEI~\eqref{eq:scaling}, one can easily verify that
\begin{equation}  
    \int \diff{s} |g_\lambda(s)|^2\langle\mathbf{:}\phi^2\mathbf{:} \rangle_\omega(\gamma(s)) \geq -\dfrac{\|g'\|^2}{8\pi^2\lambda^2} - \dfrac{\kappa^2\|g\|^2}{48\pi^2},
\end{equation}
where again $\|g\|^2$ denotes the $L^2$-norm of the function $g$. Taking the limit $\lambda \to \infty$ yields the following formula,
\begin{equation}\label{eq:longtimeWick}
 \liminf_{\lambda\xrightarrow{}\infty}\int_{-\infty}^\infty \diff{s} |g_\lambda(s)|^2\langle\mathbf{:}\phi^2\mathbf{:} \rangle_\omega(\gamma(s)) \geq - \dfrac{\kappa^2}{48\pi^2}
\end{equation}
when considering the functions $g$ such that $\|g\|^2 = 1$. Physically, since one can interpret $12\langle \mathbf{:}\phi^2\mathbf{:}\rangle$ as the square of a local temperature~\cite{BuchholzSchlemmer:2007}, 
states with negative expected Wick square are regarded
as being locally out of equilibrium. The above bound
therefore quantifies the extent to which the thermal interpretation may fail uniformly along these worldlines,
in terms of their proper acceleration. This raises an intriguing question as to whether there are states that would saturate this bound -- something quite relevant to the Unruh experiments discussed in Section~\ref{sec:discussion}.

In relation to the Unruh effect, a study of the detailed balance temperature obtained from the excitation of an Unruh-DeWitt detector carried along stationary worldlines can be found in~\cite{GoodJuarezAubryetal:2020}. Here the quantum field is assumed to be in the vacuum state, and the temperature depends not only on the curvature invariants but also on the energy gap of the detector. Although this is a different focus from our results, which concern averages of the Wick square in arbitrary Hadamard states, there are technical similarities, because the pulled back vacuum Wightman function plays a key role in both.  It would be interesting to understand whether some of the methods described here can be used to corroborate the numerical results of~\cite{GoodJuarezAubryetal:2020}.

\section{Computation of the renormalised stress-tensor for thermal and ground states on Rindler spacetime}
\label{sec:Rindler}

The Feynman propagator for a thermal state at inverse temperature $\beta$ of the massless scalar field in Minkowski spacetime was given by Dowker~\cite{Dowker:1978} and the Wightman functions (including for higher spin) can be found in~\cite{MorettiVanzo:1996}. Adopting coordinates 
$t=\xi\sinh \chi$, $x=\xi\cosh \chi$, 
the Rindler wedge $x>|t|$ of Minkowski spacetime has metric $\xi^2\diff{\chi}^2 -\diff{\xi}^2-\diff{y}^2-\diff{z}^2$,
and any curve $\chi\mapsto (a\chi,1/a,y_0,z_0)$ with $a>0$ is a curve of proper acceleration $a$ in proper time parameterisation. Given two points $x=(\chi,\xi,y,z)$ and $x'=(\chi',\xi',y',z')$, write
\begin{equation}
    \alpha(x,x') = \cosh^{-1}\left( \frac{\xi^2+(\xi')^2+(y-y')^2+(z-z')^2}{2\xi\xi'}\right),
\end{equation}
whereupon the Wightman function $G_\beta(x,x')=\langle\phi(x)\phi(x')\rangle_\beta$ for the temperature $\beta^{-1}$ KMS state with respect to the coordinate $\chi$ is
\begin{equation}
    G_\beta(x,x') = \frac{1}{4\pi\beta\xi\xi' \sinh\alpha(x,x')}\left(\frac{\sinh (2\pi\alpha(x,x')/\beta)}{\cosh(2\pi\alpha(x,x')/\beta)- \cosh(2\pi(\chi-\chi'-i\epsilon)/\beta)}\right).
\end{equation}
The $\beta=2\pi$ case coincides with the restriction of the Minkowski vacuum state to the wedge, while the zero temperature limit has Wightman function
\begin{equation}
    G_\infty(x,x') =- \frac{\alpha(x,x')}{4\pi^2\xi\xi' \sinh\alpha(x,x') (\alpha(x,x')^2 - (\chi-\chi'-i\epsilon)^2)} .
\end{equation}
To obtain the renormalised (minimally coupled) stress-energy tensor, we first apply suitable derivatives to $G_\beta-G_{2\pi}$ and take the limit $x'\to x$, obtaining
\begin{equation}
    \langle {:}(\nabla_\mu\phi)(x) (\nabla_\nu\phi)(x){:}\rangle_\beta =
    \frac{4\pi^2-\beta^2}{1440\pi^2\beta^4 \xi^4}\left(
   (16\pi^2+14\beta^2)\hat{u}_\mu \hat{u}_\nu + 30\beta^2 \hat{a}_\mu \hat{a}_\nu -(4\pi^2+11\beta^2)\eta_{\mu\nu}\right),
\end{equation}
where, at spacetime position $x$, $\hat{u}^\mu=\xi^{-1}(\partial_\chi)^\mu$ is the $4$-velocity of the curve through $x$ with constant $\xi$, $y$ and $z$, and $\hat{a}^\mu=(\partial_\xi)^\mu$ is the unit spacelike vector parallel to the $4$-acceleration of this curve. Consequently,
\begin{equation}
\langle {:}T_{\mu\nu}{:}\rangle_\beta =
\frac{4\pi^2-\beta^2}{1440\pi^2\beta^4 \xi^4}\left(
   (16\pi^2+14\beta^2)\hat{u}_\mu \hat{u}_\nu + 30\beta^2 \hat{a}_\mu \hat{a}_\nu -(4\pi^2-19\beta^2)\eta_{\mu\nu}\right)
\end{equation}
and the result for Rindler ground state is obtained by taking $\beta\to\infty$, giving
\begin{equation}
\langle {:}T_{\mu\nu}{:}\rangle_\infty =-
\frac{1}{1440\pi^2  \xi^4}\left(
    14 \hat{u}_\mu \hat{u}_\nu + 30  \hat{a}_\mu \hat{a}_\nu +19\eta_{\mu\nu}\right).
\end{equation}
Computing the energy density on curves of constant $\xi$ yields~\eqref{eq:Rindler_beta}.

\end{document}